\documentclass[aps,prd,twocolumn,superscriptaddress,amsmath,amssymb]{revtex4-2}
\usepackage{orcidlink}
\usepackage{hyperref}
\newtheorem{proposition}{Proposition}
\newenvironment{proof}{\par\noindent\textit{Proof.} }{\hfill$\square$\par}
\hypersetup{colorlinks=true,linkcolor=blue,citecolor=blue,urlcolor=blue}
\allowdisplaybreaks
\begin{document}

\title{Temporal-Plane Carroll--Schr\"odinger Dynamics and Vortex Sectors in \texorpdfstring{\((2,2)\)}{(2,2)} Klein Space}

\author{Jos\'e Rojas\,\orcidlink{0009-0002-9869-767X}}
\affiliation{Instituto de F\'isica, Universidad Aut\'onoma de Santo Domingo,
Av.\ Alma Mater, Santo Domingo 10105, Dominican Republic}

\author{Melvin Arias\,\orcidlink{0000-0001-6014-3722}}
\email{melvin.arias@intec.edu.do}
\affiliation{Instituto de F\'isica, Universidad Aut\'onoma de Santo Domingo,
Av.\ Alma Mater, Santo Domingo 10105, Dominican Republic}
\affiliation{Laboratorio de Nanotecnolog\'ia, \'Area de Ciencias B\'asicas y
Ambientales, Instituto Tecnol\'ogico de Santo Domingo, Av.\ Los
Pr\'oceres, Santo Domingo 10602, Dominican Republic}

\begin{abstract}
Motivated by the temporal dynamics identified in the \(1+1\) Carroll--Schr\"odinger theory, we derive a post-Carrollian Schr\"odinger dynamics in flat Klein space with signature \((2,2)\). Starting from the tachyonic Klein--Gordon equation in double-polar coordinates and removing a spacelike carrier, the spatial radius behaves as an effective evolution parameter, whereas the temporal two-plane \((t_1,t_2)\) serves as the equal-radius configuration space. The additional time direction supplies an \(SO(2)\) temporal angular momentum \(J\), produces temporal vortex sectors, and gives the centrifugal contribution to the post-Carrollian momentum \(P_{\rm PC}=E_\tau^2/(2M_{\rm eff})+J^2/(2M_{\rm eff}\tau^2)\) in the Hamilton--Jacobi limit. We determine the regular Bessel modes, Gaussian packets, oscillator spectrum, radial \(SU(1,1)\) tower, equal-\(r\) continuity equation, \(\mathfrak{sch}(2)\) symmetry algebra, radial-ordered propagator, and the metaplectic organization of the quadratic sectors. Effective flat connections on the temporal configuration plane give Aharonov--Bohm, Landau, and Fock--Darwin analogues, while the two-body relative sector admits anyonic boundary conditions on the punctured temporal plane. As a curved extension, we derive a branch-dependent carrier reduction and apply it to an illustrative \(SO(2,1)\)-symmetric Kleinian Schwarzschild exterior, where the Kleinian gravitational source produces a lensing-type angular deviation on the temporal plane.
\end{abstract}

\maketitle

\section{Introduction}
\label{sec:intro}

The Carrollian limit, obtained as the \(c\to0\) contraction of the Poincar\'e group, has attracted renewed attention in recent years because of its applications in gravity, holography, cosmology, hydrodynamics, condensed matter physics, and related areas of high-energy theory~\cite{LevyLeblond1965,SenGupta1966,Bacry1970,Duval2014,deBoerStories2023,deBoer2022,BagchiBjorken,CarrollHydro,Duval2022,HansenObersOlingSogaard2022,CiambelliJaiAkson2025Foundations,Ruzziconi2026CarrollianHolography,BagchiKaleidoscope2026,TadrosKolar2024CarrollBlackHoles,Bidussi2022Fractons,BagchiFlatBands2023,FigueroaOFarrillPerezProhazka2023}. In the strict Carrollian regime, admissible dynamical equations become highly constrained~\cite{deBoer2022,Banerjee2023,Marsot2022,Henneaux1979}. It is therefore useful to distinguish the strict Carrollian limit from the post-Carrollian regime, where one retains the first subleading terms around the Carroll point~\cite{EckerGrumillerSalgadoRebolledo2025,Najafizadeh2025-1}. In this regime spatial evolution is no longer completely frozen, and one can obtain Schr\"odinger-type quantum dynamics.

The Carroll--Schr\"odinger equation is a particular realization of this post-Carrollian sector. It can be derived from the tachyonic Klein--Gordon equation by factoring out a rapidly oscillating spatial carrier and then taking the Carrollian scaling limit~\cite{Najafizadeh2025-1,Najafizadeh2025-2}, either by taking $c\to \epsilon c$ with $\epsilon\to0,$ while $x$ and $t$ are fixed, or by keeping $c=1$ and $x$ fixed while taking $t\to\epsilon t$ as $\epsilon \to 0$~\cite{PhysRevD.108.046019}; in either description, the limit is performed with the rescaling $\mu \to \mu/\epsilon^2$~\cite{Najafizadeh2025-2}, which, as will be shown in this work, corresponds to and may be interpreted physically as taking the Carrollian limit while keeping the post-Carrollian momentum finite, namely,

\begin{equation}
P_{\mathrm{PC}}=\frac{mc^{3}}{2v^{2}}=\frac{E^{2}}{2mc^{3}}.
\end{equation}
In \(1+1\) dimensions this procedure produces the Carroll--Schr\"odinger equation, in which the spatial coordinate \(x\) behaves as the evolution parameter and the temporal variable as the configuration coordinate on equal-\(x\) slices~\cite{Najafizadeh2025-1,Najafizadeh2025-2,Rojas2025,RojasArias2025ManyBody}.

The resulting \(1+1\) theory has nontrivial temporal dynamics. It admits an equal-\(x\) Hilbert-space formulation, unitary spatial evolution, exact Gaussian solutions, a continuity equation, a controlled classical limit, and a propagator description~\cite{Najafizadeh2025-2,Rojas2025}. These structures were later extended to the many-body sector, where one obtains temporal interactions and nonlinear effective descriptions~\cite{RojasArias2025ManyBody}. The emergence of a distinguished temporal dynamics in the one-dimensional theory naturally motivates the question of what happens when additional temporal degrees of freedom become available.

Klein space with signature \((2,2)\) provides a natural setting for this question. Klein space arises by analytic continuation of ordinary Lorentzian spacetime and has become an important setting in both quantum field theory and general relativity~\cite{ChenHuMao2025Klein,ChenHuMao2025General,AtanasovCelestialTorus2021,CrawleyGuevaraMillerStrominger2022,EassonPezzelle2024}. In particular, \((2,2)\) signature plays a special role in canonical quantization with two time directions, scattering amplitudes, self-dual structures, and split-signature gravitational solutions~\cite{ChenHuMao2025Klein,ChenHuMao2025General,AtanasovCelestialTorus2021,DuaryMaji2024KleinSpectral,OoguriVafa1990,CrawleyGuevaraMillerStrominger2022,EassonPezzelle2024}. Flat \(\mathbb K^{2,2}\) admits a double-polar decomposition into a temporal two-plane and a spatial two-plane, making it well suited for a post-Carrollian generalization in which the temporal sector is treated dynamically~\cite{ChenHuMao2025Klein,CrawleyGuevaraMillerStrominger2022}. The central question is then whether the second temporal degree of freedom produces any genuinely new quantum, classical, and geometric structures, and how they should be interpreted.
Theories with more than one temporal direction have also appeared in other contexts, most notably in the two-time physics program of Bars and collaborators. In that approach one starts from a parent system with an additional time and a local \(Sp(2,\mathbb R)\) phase-space gauge symmetry, and different gauge choices lead to different ordinary one-time systems~\cite{Bars1998TwoTimePhysics,Bars2000FieldTheory,Bars2008Gravity2T}. Recently, this framework has also been applied to Carroll particles by choosing gauge fixings whose one-time description reproduces Carroll dynamics~\cite{KamenshchikMuscolino2024Carroll2T,KamenshchikMarraniMuscolino2025}. In the present work, the additional temporal degree of freedom has a different geometric origin, coming from the temporal plane of Klein space and surviving the
post-Carrollian limit as part of the effective equal-radius configuration space.

The results of this paper are as follows. In Sec.~\ref{sec:derivation} we derive the \(2+2\) Carroll--Schr\"odinger equation from the tachyonic Klein--Gordon equation in flat Klein space. In Sec.~\ref{sec:solutions} we analyze the regular Bessel modes, the canonical Cartesian and radial temporal variables, Gaussian spreading, temporal vortex states, the temporal oscillator, its \(SO(2)\times SU(1,1)\) organization, temporal Aharonov--Bohm flux, and magnetic-like temporal connection sectors. In Sec.~\ref{sec:continuity} we derive the effective equal-\(r\) continuity equation and identify the conserved Hilbert-space norm. In Sec.~\ref{sec:symmetry} we construct the \(\mathfrak{sch}(2)\) symmetry algebra, its Bargmann--Eisenhart interpretation, and the metaplectic organization of the quadratic sectors. In Sec.~\ref{sec:propagator} we give the radial quantization and radial-ordered propagator, emphasizing its radial initial-value structure. In Sec.~\ref{sec:classical-limit} we derive the Hamilton--Jacobi limit, the centrifugal correction to the post-Carrollian momentum, and the temporal-plane trajectories. Finally, in Sec.~\ref{sec:gravitational-sector} we derive the curved carrier reduction on split-signature Einstein backgrounds and specialize it to an illustrative \(SO(2,1)\)-symmetric Kleinian Schwarzschild branch.

\section{Derivation of the \(2+2\)
Carroll--Schr\"odinger equation}
\label{sec:derivation}

\begingroup
\setlength{\abovedisplayskip}{3pt}
\setlength{\belowdisplayskip}{3pt}
\setlength{\abovedisplayshortskip}{2pt}
\setlength{\belowdisplayshortskip}{2pt}
We begin with flat Klein space $\mathbb{K}^{2,2}$, whose metric in Cartesian coordinates is
\begin{equation}
ds^2=-(dx^0)^2-(dx^1)^2+(dx^2)^2+(dx^3)^2 .
\label{eq:klein-cartesian}
\end{equation}
A natural parametrization of $\mathbb{K}^{2,2}$ is obtained by decomposing it into a temporal two-plane and a spatial two-plane. Following the polar decomposition used in recent studies of quantum field theory in Klein space~\cite{ChenHuMao2025Klein}, we write
\begin{equation}
\begin{aligned}
(x^0,x^1)&=(c\tau\cos\psi,c\tau\sin\psi),\\
(x^2,x^3)&=(r\cos\theta,r\sin\theta).
\end{aligned}
\label{eq:polar-decomp}
\end{equation}
where $\tau\geq0$ is the radial temporal coordinate and $\psi\in[0,2\pi)$ is the temporal angle, while $r$ and $\theta$ are the polar coordinates in the spatial plane; hence
\begin{equation}
ds^2=-c^2d\tau^2-c^2\tau^2d\psi^2+dr^2+r^2d\theta^2 .
\label{eq:k2_2_metric}
\end{equation}
For a scalar field $\phi$ on this background, with signature $(-,-,+,+)$, we use the density
\begin{equation}
\mathcal L=
\frac12\sqrt{|g|}\left(g^{\mu\nu}\partial_\mu\phi\,\partial_\nu\phi-\mu^2\phi^2\right),
\label{eq:flat-tachyonic-lagrangian}
\end{equation}
which gives the tachyonic Klein--Gordon equation
\begin{equation}
(\Box_g+\mu^2)\phi=0,
\qquad
\mu=\frac{mc}{\hbar} .
\label{eq:kg-convention}
\end{equation}
\endgroup
For the metric \eqref{eq:k2_2_metric}, this gives the equivalent double-polar form
\begin{multline}
\frac{1}{c^2\tau}\partial_\tau(\tau\partial_\tau\phi)
+\frac{1}{c^2\tau^2}\partial_\psi^2\phi
-\mu^2\phi
\\
-\frac{1}{r}\partial_r(r\partial_r\phi)
-\frac{1}{r^2}\partial_\theta^2\phi=0 .
\label{eq:kg-klein}
\end{multline}

To extract the post-Carrollian sector, we follow the same logic as in the one-dimensional construction and factor out a rapidly varying spatial phase~\cite{Najafizadeh2025-2},
\begin{equation}
\phi(\tau,\psi,r,\theta)=\frac{e^{-i\mu r}}{\sqrt{\mu}}\,\widetilde\phi(\tau,\psi,r,\theta).
\label{eq:phase-factor}
\end{equation}
Here we have used the radially symmetric carrier-wave factorization; however, an asymmetric factorization may also be taken along a fixed spatial direction, $\vec{P}=p\hat n$, where $\hat n$ denotes the unit vector in the momentum direction, a choice that makes the effective role of space as the driver of evolution more explicit.
 Substituting Eq.~\eqref{eq:phase-factor} into Eq.~\eqref{eq:kg-klein}, dividing by the common carrier factor, and canceling the mass term against the leading contribution from $\partial_r^2\phi$, one obtains
\begin{multline}
\frac{1}{c^2}\left[
\frac{1}{\tau}\partial_\tau(\tau\partial_\tau\widetilde\phi)
+
\frac{1}{\tau^2}\partial_\psi^2\widetilde\phi
\right]
\\
-
\left[
\partial_r^2\widetilde\phi
+\frac{1}{r}\partial_r\widetilde\phi
-2i\mu\partial_r\widetilde\phi
-\frac{i\mu}{r}\widetilde\phi
+\frac{1}{r^2}\partial_\theta^2\widetilde\phi
\right]=0 .
\label{eq:pre-carroll}
\end{multline}

We now take the Carrollian scaling appropriate to the post-Carrollian Carroll--Schr\"odinger sector. The rescaling is imposed on the temporal Cartesian coordinates while keeping spatial variables fixed, which leads to
\begin{equation}
c\tau\longrightarrow \epsilon c\tau,
\label{eq:temporal-contraction}
\end{equation}
together with the rescaling $\mu\to\mu/\epsilon^2$~\cite{Najafizadeh2025-2} this gives
\begin{multline}
\frac{\hbar^2}{2M_{\rm eff}}
\left[
\frac{1}{\tau}\partial_\tau(\tau\partial_\tau\widetilde\phi)
+
\frac{1}{\tau^2}\partial_\psi^2\widetilde\phi
\right]
+i\hbar\left(\partial_r+\frac{1}{2r}\right)\widetilde\phi
\\
-
\epsilon^2\frac{\hbar^2}{2mc}
\left(
\partial_r^2+\frac{1}{r}\partial_r+\frac{1}{r^2}\partial_\theta^2
\right)\widetilde\phi
=0 .
\label{eq:prelimit-scaled-envelope}
\end{multline}
Taking the limit $\epsilon\to0$, one sees that the $\theta$ dependence and second-order radial-derivative terms in Eq.~\eqref{eq:prelimit-scaled-envelope} disappear, leading to the Carroll--Schr\"odinger equation:
\begin{equation}
\boxed{
\frac{\hbar^2}{2M_{\rm eff}}
\left[
\frac{1}{\tau}\partial_\tau(\tau\partial_\tau\widetilde\phi)
+
\frac{1}{\tau^2}\partial_\psi^2\widetilde\phi
\right]
+i\hbar\left(\partial_r+\frac{1}{2r}\right)\widetilde\phi=0 .}
\label{eq:22CS-final}
\end{equation}
The $\mu\to\mu/\epsilon^2$ contraction while taking the $c\tau\to\epsilon c\tau$ limit used to derive the Carroll--Schr\"odinger equation may be interpreted physically as selecting the sector in which the post-Carrollian momentum
\begin{equation}
P_{\rm PC}=\frac{mc^3}{2v^2}=\frac{E^2}{2mc^3}
\label{eq:PPC-derivation-section}
\end{equation}
remains finite, as reviewed in Appendix~\ref{app:pc-scaling}. We use the representative \(c\to\epsilon c\) with all coordinates kept fixed, for which the finite combination
\begin{equation}
M_{\rm eff}=mc^3
\label{eq:Meff-main}
\end{equation}
is manifest and the resulting equation \eqref{eq:22CS-final} takes a Schr\"odinger-like form. As will be seen in this work, the resulting Carroll--Schr\"odinger equation describes the effective radial evolution in the spatial variable \(r\), while the temporal sector is governed by the radial temporal coordinate \(\tau\) and the temporal angle \(\psi\). The effective \(1/r\) term originates from the polar carrier-wave factorization in Eq.~\eqref{eq:phase-factor}; a more general treatment is given in Appendix~\ref{app:generalized-potential} and in Sec.~\ref{sec:gravitational-sector}. It is useful to note that under the half-density redefinition $\widetilde\phi=r^{-1/2}\Phi$, the $1/(2r)$ term cancels and $\Phi$ satisfies the standard two-dimensional free Schr\"odinger equation with $r$ playing the role of time. Since \(P_{\mathrm{PC}}\) is the quantity kept finite in the contraction, Eq.~\eqref{eq:22CS-final} has been written directly in momentum units and will be the form used in the rest of the paper.

\section{Solutions}
\label{sec:solutions}

\subsection{Free modes and temporal vortex states}
\label{subsec:free-modes}

We study the free solutions of Eq.~\eqref{eq:22CS-final}. With the separated ansatz
\begin{equation}
\widetilde\phi(\tau,\psi,r)=T(\tau)\Psi(\psi)R(r),
\label{eq:sep-ansatz-new}
\end{equation}
the angular equation gives
\begin{equation}
\Psi_n(\psi)=e^{in\psi},
\qquad n\in\mathbb Z.
\label{eq:angular-modes-new}
\end{equation}
A general angular solution is obtained from linear combinations of the modes
$e^{in\psi}$ and $e^{-in\psi}$, or equivalently from the real basis
$\cos(n\psi)$ and $\sin(n\psi)$. The radial spatial equation,
\begin{equation}
i\hbar\left(\partial_rR+\frac{R}{2r}\right)=pR,
\qquad p\geq0,
\label{eq:radial-eq}
\end{equation}
gives
\begin{equation}
R(r)=C_r r^{-1/2}\exp\left(-\frac{ipr}{\hbar}\right).
\label{eq:R-free-solution-new}
\end{equation}
For $p<0$ the temporal radial equation for $T(\tau) $ produces modified Bessel functions
$I_{|n|}$ and $K_{|n|}$, describing an evanescent temporal sector outside the
post-Carrollian momentum shell; the physical sector is $p\geq0$. With
$\lambda^2=2M_{\rm eff}p/\hbar^2$, the temporal radial equation is the Bessel
equation of order $|n|$ with solutions $J_{|n|}$ and $Y_{|n|}$, and regularity at $\tau=0$, as in Klein-space mode
analyses~\cite{ChenHuMao2025Klein}, selects
\begin{equation}
T_n^{\rm reg}(\tau)=J_{|n|}(\lambda\tau).
\label{eq:Treg-bessel}
\end{equation}
Combining the radial, angular, and temporal parts, a basis of regular separated modes is
\begin{equation}
\begin{aligned}
\widetilde\phi_{n,p}
&=\frac{\mathcal{N}_{n,p}}{\sqrt{r}}\,
e^{-ipr/\hbar}\,e^{in\psi}\,J_{|n|}(\lambda\tau),\\
\lambda&=\frac{\sqrt{2M_{\rm eff}p}}{\hbar},\qquad p\geq0 .
\end{aligned}
\label{eq:free-regular-mode}
\end{equation}

Each mode with $n\neq0$ behaves as a temporal vortex state. For the
rescaled field $\Phi=r^{1/2}\widetilde\phi$, the temporal probability current, as shown in Sec. \ref{sec:continuity}, is
\begin{subequations}
\begin{align}
J_\tau &= \frac{\hbar}{2iM_{\rm eff}}\left(\Phi^*\partial_\tau\Phi
-\Phi\,\partial_\tau\Phi^*\right),
\label{eq:Jtau}
\\
J_\psi &= \frac{\hbar}{2iM_{\rm eff}\tau}\left(\Phi^*\partial_\psi\Phi
-\Phi\,\partial_\psi\Phi^*\right).
\label{eq:Jpsi-general}
\end{align}
\end{subequations}
For the Bessel modes, $J_\tau=0$, while
\begin{equation}
J_\psi=\frac{\hbar n}{M_{\rm eff}\tau}\,|T_n(\tau)|^2 .
\label{eq:Jpsi}
\end{equation}
The associated temporal circulation is then quantized,
\begin{equation}
\Gamma_n=\oint v_\psi\,\tau\,d\psi
=\frac{2\pi\hbar n}{M_{\rm eff}},
\label{eq:temporal-vortex-circulation}
\end{equation}
where $v_\psi=J_\psi/|\Phi_n|^2=\hbar n/(M_{\rm eff}\tau)$. The circulation is
independent of $\tau$ and is a topologically invariant quantity. The temporal vorticity is
\( 
\nabla_t\times\mathbf v=2\pi\hbar n \delta^{(2)}(\mathbf t)/{M_{\rm eff}}
\), which is concentrated at the temporal origin, while the circulation remains finite and fixed by the winding number. For completeness, we note that the $n=0$ radial sector has deficiency indices $(1,1)$ and therefore admits a one-parameter family of self-adjoint extensions beyond the regular one. Correspondingly, it allows the usual logarithmic boundary behavior of the two-dimensional radial Laplacian near the origin,
\begin{equation}
T_0(\tau)\sim A+\frac{2B}{\pi}\ln\!\left(\frac{\tau}{\tau_0}\right),
\qquad \tau\to0^+ .
\label{eq:equation-for-tau}
\end{equation}
In this work we restrict to the regular sector $(B=0)$. Nonzero $B$ would correspond to a contact-type boundary condition at the temporal origin and would define a different self-adjoint extension of the $n=0$ problem.

\subsection{Canonical structure in Cartesian and radial temporal variables}
\label{subsec:ccr}

The canonical structure is simplest in Cartesian temporal coordinates
$t_a$, $a=1,2$. On the standard domain in $L^2(\mathbb{R}^2,d^2t)$ the canonical pair of the reduced theory is given by
\begin{equation}
\widehat E_a=-i\hbar\partial_{t_a},
\qquad
[\hat t_a,\widehat E_b]=i\hbar\delta_{ab}.
\label{eq:cartesian-temporal-ccr}
\end{equation}
The radial operator, in
direct analogy with the radial momentum operator in ordinary polar
coordinates~\cite{GilPaz2001},
\begin{equation}
\hat{E}_\tau=-i\hbar\!\left(\partial_\tau+\frac{1}{2\tau}\right),
\qquad
\hat{L}_\psi=-i\hbar\,\partial_\psi,
\label{eq:E-L-ops}
\end{equation}
is the polar representative of the Cartesian temporal momentum along
$\hat e_\tau$. On the common polar domain one has the formal relation
$[\hat\tau,\hat E_\tau]=i\hbar$. In the angular sector,
$\widehat L_\psi=-i\hbar\partial_\psi$ is self-adjoint on periodic functions
on $S^1$, and its spectrum is $n\hbar$, $n\in\mathbb Z$, where $n$ labels
the temporal angular momentum sectors.

\subsection{Gaussian wave packets and temporal Rayleigh range}
\label{subsec:gaussian}

The temporal spreading is more naturally described in terms of the rescaled field
\( 
\Phi(\mathbf{t},r)=\sqrt{r}\,\widetilde{\phi}(\mathbf{t},r),
\)
for which Eq.~\eqref{eq:22CS-final} becomes
\begin{equation}    
\frac{\hbar^2}{2M_{\rm eff}}\nabla_t^2\Phi+i\hbar\,\partial_r\Phi=0.
\label{eq:free-temporal-plane}
\end{equation}
In this form, \(\Phi\) may be treated as a free two-dimensional wave packet on the temporal plane. For an isotropic Gaussian initial state
\begin{equation}
\Phi_0(\mathbf{t})=\frac{1}{\sqrt{\pi}\,\sigma_0}
\exp\!\left(-\frac{|\mathbf{t}|^2}{2\sigma_0^2}\right),
\label{eq:gaussian-initial}
\end{equation}
normalized in \(L^2(\mathbb{R}^2,d^2t)\), the propagated state at \(r=r_0+\Delta r\) is
\begin{equation}
\Phi(\mathbf{t},r)=\frac{i\gamma/\beta}{\sqrt{\pi}\,\sigma_0}
\exp\!\left(-\frac{|\mathbf{t}|^2}{2\sigma(r)^2}\right),
\label{eq:gaussian-propagated}
\end{equation}
where
\begin{equation}
\sigma(r)^2=\sigma_0^2+\frac{i\hbar\,\Delta r}{M_{\rm eff}},
\quad
\gamma=\frac{M_{\rm eff}}{2\hbar\,\Delta r},
\quad
\beta=\frac{1}{2\sigma_0^2}-i\gamma.
\label{eq:gaussian-parameters}
\end{equation}
For the original ungauged envelope this gives
\begin{equation}
\widetilde{\phi}(\mathbf{t},r)=\frac{1}{\sqrt r}\,\Phi(\mathbf{t},r),
\qquad
|\widetilde{\phi}(\mathbf{t},r)|^2=\frac{1}{r}\,|\Phi(\mathbf{t},r)|^2 .
\label{eq:ecuacionjose}
\end{equation}
The temporal density of the rescaled field is
\begin{equation}
|\Phi(\mathbf{t},r)|^2=\frac{1}{\pi W(r)^2}
\exp\!\left(-\frac{|\mathbf{t}|^2}{W(r)^2}\right),
\label{eq:gaussian-density}
\end{equation}
with temporal width
\begin{equation}
W(r)^2=\sigma_0^2+\frac{\hbar^2(r-r_0)^2}{M_{\rm eff}^2\sigma_0^2}.
\label{eq:gaussian-width}
\end{equation}
Hence
\begin{equation}
\langle\tau^2\rangle(r)=W(r)^2=\sigma_0^2\!\left[1+
\left(\frac{r-r_0}{r_*}\right)^{\!2}\right],
\label{eq:temporal-spreading}
\end{equation}
where
\begin{equation}
r_*=\frac{M_{\rm eff}\sigma_0^2}{\hbar}=\frac{mc^3\sigma_0^2}{\hbar}
\label{eq:rayleigh-range}
\end{equation}
is the temporal Rayleigh range. For \(\Delta r\ll r_*\) the Gaussian is essentially undistorted, while for \(\Delta r\gg r_*\) one has \(W\sim\hbar\Delta r/(M_{\rm eff}\sigma_0)\). The Gaussian~\eqref{eq:gaussian-initial} is a minimum-uncertainty state for the Cartesian pairs \((t_a,E_a)\) at \(r=r_0\), and Eq.~\eqref{eq:temporal-spreading} shows that its temporal width increases with radial propagation.

\subsection{Temporal oscillator, \(SO(2)\times SU(1,1)\)
symmetry, and degeneracy}
\label{subsec:oscillator}

For the rescaled field $\widetilde\phi=r^{-1/2}\Phi$,
the free equation \eqref{eq:free-temporal-plane} takes the Schr\"odinger form
$i\hbar\partial_r\Phi=\hat P_0\Phi$, with
$\hat P_0=-(\hbar^2/2M_{\rm eff})\nabla_t^2$. A shift of $\hat E_\tau$ by a quadratic potential $\kappa\tau^2/2$ dresses the free Bessel modes by the phase $\exp(i\kappa\tau^3/6\hbar)$ without changing the spectrum. By contrast, a genuine oscillator arises through a momentum-sector coupling, consistent with the interaction-momentum viewpoint developed in the one-dimensional and many-body Carroll--Schr\"odinger settings~\cite{Rojas2025,RojasArias2025ManyBody},
\begin{equation}
i\hbar\partial_r\Phi=
\left[-\frac{\hbar^2}{2M_{\rm eff}}\nabla_t^2
+\frac12M_{\rm eff}\Omega^2\tau^2\right]\Phi .
\label{eq:oscillator-eq}
\end{equation}
The regular stationary modes $e^{-ip_{N,n}r/\hbar}e^{in\psi}T_{N,n}(\tau)$ are the two-dimensional isotropic oscillator modes
\begin{equation}
T_{N,n}(\tau)=\mathcal N_{N,n}\rho^{|n|/2}e^{-\rho/2}L_N^{|n|}(\rho),
\qquad
\rho=\frac{M_{\rm eff}\Omega}{\hbar}\tau^2,
\label{eq:oscillator-modes}
\end{equation}
with quantized radial momentum
\begin{equation}
p_{N,n}=\hbar\Omega(2N+|n|+1),
\qquad N=0,1,2,\ldots,\quad n\in\mathbb Z .
\label{eq:momentum-quantization}
\end{equation}
At principal level $k=2N+|n|$, the allowed values are $n=-k,-k+2,\ldots,k$, with $N=(k-|n|)/2\in\mathbb Z_{\ge0}$, giving degeneracy $k+1$. 

The operator organization is most transparent in Cartesian temporal variables. With $\widehat E_a=-i\hbar\partial_{t_a}$ and $\ell_\Omega=\sqrt{\hbar/(M_{\rm eff}\Omega)}$, define
\begin{equation}
a_a=\frac{1}{\sqrt2}\left(\frac{t_a}{\ell_\Omega}+i\frac{\ell_\Omega}{\hbar}\widehat E_a\right),
\qquad
[a_a,a_b^\dagger]=\delta_{ab} .
\label{eq:cartesian-ladders-compact}
\end{equation}
The polar-coordinate representation may be obtained by a change of variables. Then
\begin{equation}
\hat P_\Omega=\hbar\Omega(a_1^\dagger a_1+a_2^\dagger a_2+1).
\end{equation}
We introduce circular operators
\begin{equation}
a_+=\frac{1}{\sqrt2}(a_1-ia_2),
\qquad
a_-=\frac{1}{\sqrt2}(a_1+ia_2),
\label{eq:circular-ladders-compact}
\end{equation}
and number operators $N_\pm=a_\pm^\dagger a_\pm$.  One then has
\begin{equation}
\hat P_\Omega=\hbar\Omega(N_++N_-+1),
\qquad
\widehat L_\psi=\hbar(N_+-N_-).
\label{eq:circular-HL-compact}
\end{equation}
The relation with the radial labels is
\begin{equation}
\begin{gathered}
n=N_+-N_-,\qquad N=\min(N_+,N_-),\\
N_++N_-=2N+|n| .
\end{gathered}
\end{equation}
which reproduces Eq.~\eqref{eq:momentum-quantization}. The radial $SU(1,1)$ tower at fixed $n$ is generated by the bilinears
\begin{equation}
\begin{gathered}
K_+=a_+^\dagger a_-^\dagger,
\qquad
K_-=a_+a_-,\\
K_0=\frac12(N_++N_-+1).
\end{gathered}
\label{eq:su11-bilinears-compact}
\end{equation}
which satisfy $[K_0,K_\pm]=\pm K_\pm$ and $[K_+,K_-]=-2K_0$, where $K_+$ raises the radial quantum number $N$ while preserving temporal angular momentum $n$, and the lowest-weight representation at fixed $n$ has Bargmann index
\begin{equation}
k_n=\frac{|n|+1}{2} .
\label{eq:bargmann-index}
\end{equation}

\paragraph{Effective temporal Aharonov--Bohm flux.}
Within the same angular sector, one may introduce an effective flat \(U(1)\) connection on the punctured temporal configuration plane, so that the wave function acquires a holonomy phase around the puncture. The coupling is introduced through the shift
\begin{equation}
\widehat L_\psi\longrightarrow \widehat L_\psi-\alpha\hbar,
\qquad \alpha\in[0,1) .
\label{eq:AB-shift}
\end{equation}
The Bessel and oscillator indices are then shifted from $|n|$ to $|n-\alpha|$. In particular,
\begin{equation}
p_{N,n}(\alpha)=\hbar\Omega(2N+|n-\alpha|+1),
\qquad n\in\mathbb Z .
\label{eq:AB-spectrum}
\end{equation}
For generic $\alpha$ this lifts the integer-flux oscillator degeneracies; in $0<\alpha<1$ a residual doublet occurs only at $\alpha=1/2$, where $n=0$ and $n=1$ have the same $|n-\alpha|$.

\subsection{Temporal Gouy phase and metaplectic grading of vortex sectors}
\label{subsec:gouy}

The free temporal spreading of Sec.~\ref{subsec:gaussian} and the
oscillator spectrum of Sec.~\ref{subsec:oscillator} are controlled by the
same metaplectic weight. In the normalizable basis adapted to the Gaussian scale \(\sigma_0\),
with Rayleigh range \(r_*=M_{\rm eff}\sigma_0^2/\hbar\) and
\(\Delta r=r-r_0\), the solutions of the free equation
\eqref{eq:free-temporal-plane} may also be written as temporal
Laguerre--Gauss modes
\begin{multline}
\Phi_{N,n}(\mathbf t,r)=
\frac{\mathcal C_{N,n}}{W(r)}
\left(\frac{\tau}{W(r)}\right)^{|n|}
L_N^{|n|}\!\left(\frac{\tau^2}{W(r)^2}\right)
e^{in\psi}
\\
\times \exp\!\left(-\frac{\tau^2}{2\sigma(r)^2}\right)
\exp\!\left[-i(2N+|n|+1)\zeta(r)\right],
\label{eq:temporal-LG}
\end{multline}
with \(\sigma(r)^2\) and \(W(r)^2\) as in
Eqs.~\eqref{eq:gaussian-parameters} and \eqref{eq:gaussian-width}, and
Gouy phase
\begin{equation}
\zeta(r)=\arctan\!\left(\frac{r-r_0}{r_*}\right).
\label{eq:gouy-phase}
\end{equation}
The fundamental \(N=n=0\) mode is the Gaussian of
Sec.~\ref{subsec:gaussian}, whose axial phase relative to the focus is
\(-\zeta(r)\).

Across the temporal focus, \(\zeta(r)\) runs from \(-\pi/2\) to \(+\pi/2\),
so the total Gouy shift is \(\pi\). The phase is graded by the temporal
angular momentum,
\begin{equation}
\Theta_{N,n}(r)=(2N+|n|+1)\,\zeta(r),
\label{eq:graded-gouy}
\end{equation}
so vortex sectors of different \(|n|\) dephase relative to one another as
the state propagates in \(r\).

The same integer weight also appears in the confined problem. At fixed
oscillator scale \(\ell_\Omega=\sqrt{\hbar/(M_{\rm eff}\Omega)}\), the
temporal oscillator \eqref{eq:oscillator-eq} is diagonalized by
Laguerre--Gauss profiles with radial and angular labels \((N,n)\), which
evolve with the linear phase
\(\exp[-i\,p_{N,n}\Delta r/\hbar]
=\exp[-i\,\Omega(2N+|n|+1)\Delta r]\).
Hence
\begin{equation}
\frac{p_{N,n}}{\hbar\Omega}=2N+|n|+1,
\label{eq:gouy-spectrum-identity}
\end{equation}
so the same weight \(2N+|n|+1\) governs both the free Gouy phase and the
oscillator spectrum. In the free problem it appears through the running
\(\zeta(r)\), whereas in the oscillator it appears through the linear phase
\(\Omega\Delta r\). The discrete spectrum
\eqref{eq:momentum-quantization} and the geometric phase of the freely
spreading packet are therefore controlled by the same weight
\(2N+|n|+1\), realized on two one-parameter subgroups of
\(Mp(4,\mathbb R)\). In this language the half-integer offsets of the
metaplectic representation may be viewed as Maslov-type contributions
associated with the temporal focus, consistent with the total Gouy shift
\(\pi\) on the two-dimensional plane.

This grading can be seen in a two-mode superposition. For the equal-\(r\)
state \(\Phi=c_0\Phi_{0,0}+c_n\Phi_{0,n}\) prepared at the focus, free
radial propagation produces the relative phase
\begin{equation}
\delta\Theta(r)=|n|\,\arctan\!\left(\frac{r-r_0}{r_*}\right),
\label{eq:two-mode-dephasing}
\end{equation}
which saturates at \(|n|\pi/2\) in the far temporal field. Since
\(\Phi_{0,n}\) carries azimuthal index \(n\), this dephasing appears as a
rigid rotation of the temporal interference pattern by the angle
\begin{equation}
\Delta\psi(r)=\frac{\delta\Theta(r)}{n}
=\mathrm{sgn}(n)\arctan\!\left(\frac{r-r_0}{r_*}\right),
\qquad n\neq0,
\label{eq:lobe-rotation}
\end{equation}
which grows monotonically with the radial evolution parameter and depends on
the geometry only through \(r_*\). This effect has no counterpart in the
\(1+1\) Carroll--Schr\"odinger theory, where there is no temporal angular
momentum to grade the phase.

\subsection{Effective temporal connection sector and Fock--Darwin-type radial spectrum}
\label{subsec:temporal-magnetic}

The two-dimensional temporal plane also permits an effective connection with nonzero curvature. We use this sector only as a planar model coupling on the reduced configuration space. Let
\begin{equation}
\boldsymbol{\Pi}_t=-i\hbar\nabla_t-\mathbf A_t(t_1,t_2),
\qquad
\mathcal B_t=\partial_{t_1}A_{t_2}-\partial_{t_2}A_{t_1}.
\end{equation}
Here \(\mathcal B_t\) denotes the planar curvature of the effective connection. It should not be identified with the magnetic induction of ordinary electromagnetism. A first-principles electromagnetic treatment would require starting from a Carrollian electrodynamics sector~\cite{deBoerStories2023,PhysRevD.108.046019} and extending that construction to the post-Carrollian regime considered here. The radial rescaled equation
\begin{equation}
i\hbar\partial_r\Phi=\frac{1}{2M_{\rm eff}}\boldsymbol{\Pi}_t^2\Phi
\label{eq:temporal-landau-eq}
\end{equation}
defines a Landau-type sector. For a uniform temporal curvature, choose the symmetric gauge
\begin{equation}
\mathbf A_t=\frac{\mathcal B_t}{2}(-t_2,t_1),
\qquad
\omega_t=\frac{|\mathcal B_t|}{M_{\rm eff}}.
\end{equation}
Then Eq.~\eqref{eq:temporal-landau-eq} has Landau-type radial momenta
\begin{equation}
p_N^{\rm L}=\hbar\omega_t\left(N+\frac12\right),
\qquad N=0,1,2,\ldots,
\label{eq:temporal-landau-levels}
\end{equation}
with degeneracy controlled by the temporal area in the usual planar
way. In the same sector one may define guiding-center coordinates
whose commutator is
\begin{equation}
[R_1,R_2]=i\ell_t^2,
\qquad
\ell_t^2=\frac{\hbar}{|\mathcal B_t|}.
\label{eq:guiding-center}
\end{equation}
Combining the uniform temporal curvature with the oscillator gives the
Fock--Darwin Hamiltonian
\begin{equation}
i\hbar\partial_r\Phi=\left[\frac{\boldsymbol{\Pi}_t^2}{2M_{\rm eff}}
+\frac12M_{\rm eff}\Omega^2\tau^2\right]\Phi.
\label{eq:fock-darwin-temporal}
\end{equation}
It diagonalizes into two chiral oscillators with frequencies
\begin{equation}
\Omega_\pm=\sqrt{\Omega^2+\frac{\omega_t^2}{4}}\pm\frac{\omega_t}{2},
\label{eq:Omega-pm}
\end{equation}
and spectrum
\begin{equation}
\begin{aligned}
p_{n_+,n_-}
&=\hbar\Omega_+\left(n_++\frac12\right)
+\hbar\Omega_-\left(n_-+\frac12\right),\\
n_\pm&=0,1,2,\ldots .
\end{aligned}
\label{eq:fock-darwin-spectrum}
\end{equation}
For $\mathcal B_t\to0$ this reduces to the isotropic temporal oscillator;
for $\Omega\to0$ one chiral frequency vanishes and the Landau degeneracy
is recovered.

\subsection{Two-body relative sector and anyonic temporal holonomy}
\label{subsec:anyons}

The single-particle temporal Aharonov--Bohm sector of Eq.~\eqref{eq:AB-shift}
has a two-body analogue. For two excitations the relative temporal coordinate after the coincidence point is removed
is \(\mathbb R^2\setminus\{0\}\). A flat holonomy around this puncture then gives
a family of anyonic boundary conditions for the relative wave function, which is
how the second temporal direction enters the two-body sector.

Consider two identical post-Carrollian excitations described, in the many-body
extension of the reduced theory~\cite{RojasArias2025ManyBody}, by the free
two-body equation for the rescaled field \(\Phi(\mathbf t_1,\mathbf t_2,r)\),
\begin{equation}
i\hbar\,\partial_r\Phi
=
-\frac{\hbar^2}{2M_{\rm eff}}
\left(
\nabla_{t_1}^2+\nabla_{t_2}^2
\right)\Phi .
\label{eq:two-body-free}
\end{equation}
Introducing temporal center-of-mass and relative coordinates
\(\mathbf T=(\mathbf t_1+\mathbf t_2)/2\) and
\(\mathbf t=\mathbf t_1-\mathbf t_2\), one obtains
\begin{equation}
\nabla_{t_1}^2+\nabla_{t_2}^2
=
\frac{1}{2}\nabla_T^2+2\nabla_t^2 .
\label{eq:two-body-laplacian-split}
\end{equation}
With \(\Phi=\Phi_{\rm cm}(\mathbf T,r)\chi(\mathbf t,r)\), the relative motion
obeys
\begin{equation}
i\hbar\,\partial_r\chi
=
-\frac{\hbar^2}{2\mu_{\rm rel}}\nabla_t^2\chi,
\qquad
\mu_{\rm rel}=\frac{M_{\rm eff}}{2},
\label{eq:relative-eq}
\end{equation}
on \(\mathbf t\in\mathbb R^2\setminus\{0\}\). The removal of the coincidence
point is the usual idealization behind pointlike anyonic quantization, in which a microscopic interaction or core regularization would select a definite
self-adjoint extension.

\begin{proposition}[Anyonic temporal holonomy]
\label{prop:anyons}
The relative sector \eqref{eq:relative-eq} admits a \(U(1)\) family of
inequivalent boundary conditions labeled by \(\alpha\in[0,2)\), implemented by
the flat temporal Aharonov--Bohm connection
\(\mathbf a=\alpha\hbar\nabla\psi\) on the relative angle. In a confining
temporal oscillator of relative frequency \(\Omega_{\rm rel}\), the relative
radial-momentum spectrum is
\begin{equation}
\begin{aligned}
p^{\rm rel}_{N,m}(\alpha)
&=
\hbar\Omega_{\rm rel}
\bigl(2N+|m-\alpha|+1\bigr),\\
N&\in\mathbb Z_{\ge0},
\qquad
m\in\mathbb Z .
\end{aligned}
\label{eq:anyon-spectrum}
\end{equation}
A full loop around the coincidence point gives the monodromy
\(e^{2\pi i\alpha}\). The exchange path corresponds to the half-rotation
\(\mathbf t\to-\mathbf t\) and gives the phase
\begin{equation}
\mathcal E(\alpha)=e^{i\pi\alpha}.
\label{eq:exchange-phase}
\end{equation}
The case \(\alpha=0\) is bosonic, while \(\alpha\equiv1\pmod{2}\) gives the
fermionic exchange sign; generic \(\alpha\) gives an anyon-like temporal
holonomy.
\end{proposition}

\begin{proof}
Writing \(\mathbf t=\tau(\cos\psi,\sin\psi)\), the relative configuration space
is the punctured plane, with
\begin{equation}
\pi_1(\mathbb R^2\setminus\{0\})=\mathbb Z .
\label{eq:punctured-plane-pi1}
\end{equation}
Its one-dimensional unitary representations are labeled by the holonomy
\(e^{2\pi i\alpha}\) around the puncture~\cite{LeinaasMyrheim1977,Wilczek1982}.
The flat connection \(\mathbf a=\alpha\hbar\nabla\psi\) realizes this holonomy
by shifting the relative angular momentum as
\begin{equation}
\widehat L_\psi\longrightarrow \widehat L_\psi-\alpha\hbar .
\label{eq:anyon-angular-shift}
\end{equation}
The half-loop \(\psi\to\psi+\pi\) gives Eq.~\eqref{eq:exchange-phase}, and the
oscillator spectrum follows from the corresponding replacement
\(|m|\to|m-\alpha|\).
\end{proof}

\section{Continuity equation and equal-\texorpdfstring{$r$}{r} norm}
\label{sec:continuity}

Multiplying Eq.~\eqref{eq:22CS-final} by $\widetilde\phi^*$,
subtracting the complex conjugate, and rearranging gives
\begin{equation}
\frac{1}{r}\partial_r(r|\widetilde\phi|^2)+\nabla_t\cdot\mathbf{J}_{\widetilde\phi}=0,
\label{eq:continuity-tilde}
\end{equation}
with temporal current
$\mathbf{J}_{\widetilde\phi}=J_\tau\hat e_\tau+J_\psi\hat e_\psi$
given by:
\begin{subequations}
\begin{align}
J_\tau&=\frac{\hbar}{2iM_{\rm eff}}\!
\left(\widetilde\phi^*\partial_\tau\widetilde\phi
-\widetilde\phi\,\partial_\tau\widetilde\phi^*\right),
\\
J_\psi&=\frac{\hbar}{2iM_{\rm eff}\tau}\!
\left(\widetilde\phi^*\partial_\psi\widetilde\phi
-\widetilde\phi\,\partial_\psi\widetilde\phi^*\right).
\end{align}
\end{subequations}
Under $\widetilde\phi=r^{-1/2}\Phi$, the identity $r|\widetilde\phi|^2=|\Phi|^2$
eliminates the polar measure term and gives
$\partial_r|\Phi|^2+\nabla_t\cdot\mathbf{J}_\Phi=0$.
Integrating over the temporal measure
$d\Sigma_t=\tau\,d\tau\,d\psi$, the equal-$r$ norm
\begin{equation}
\mathcal{N}(r)=\int_0^\infty\!\tau\,d\tau
\int_0^{2\pi}\!d\psi\,|\Phi(\tau,\psi,r)|^2,
\qquad\frac{d\mathcal{N}}{dr}=0,
\label{eq:conserved-norm}
\end{equation}
is conserved.  The Hilbert space is
$L^2(\mathbb{R}^+\times S^1,\tau\,d\tau\,d\psi)\simeq L^2(\mathbb{R}^2,d^2t)$.
For the original ungauged field the same conserved quantity is written as
\begin{equation}
\widetilde{\mathcal N}(r)=r\int_0^\infty\!\tau\,d\tau\int_0^{2\pi}\!d\psi\,
|\widetilde\phi(\tau,\psi,r)|^2=\mathcal N(r).
\label{eq:conserved-norm-tilde}
\end{equation}
The general regular state admits the Bessel--Fourier spectral resolution, which makes the radial unitarity explicit:
\begin{align}
\Phi(r,\tau,\psi)
={}&\sum_{n\in\mathbb{Z}}\int_0^\infty\!d\lambda\,
a_n(\lambda)\nonumber\\
&\times e^{-i\hbar\lambda^2 r/(2M_{\rm eff})}
 e^{in\psi}\,J_{|n|}(\lambda\tau),
\label{eq:spectral-resolution}
\end{align}
where the Bessel measure is absorbed in $a_n(\lambda)$.  
More generally, if X is an anti-Hermitian symmetry of the equal-r equation, preserving the chosen domain, then \(-i\hbar X\)      defines the conserved Hermitian charge
\begin{equation}
Q_X(r)=\int_{\mathbb R^2}d^2t \Phi^\dagger(-i\hbar X)\Phi .
\label{eq:noether-charge}
\end{equation}

\section{Symmetry algebra}
\label{sec:symmetry}

Starting from the gauge-fixed reduced equation, in close analogy with the one-dimensional Carroll--Schr\"odinger case~\cite{Najafizadeh2025-2},
\begin{equation}
\left(i\hbar\partial_r+\frac{\hbar^2}{2M_{\rm eff}}\Delta_t\right)\Phi=0,
\qquad
\Delta_t=\partial_{t_1}^2+\partial_{t_2}^2,
\label{eq:gauge-fixed-symmetry-start}
\end{equation}

one obtains the generators in Cartesian coordinates
\begin{subequations}
\begin{align}
P&=\partial_r,\quad
H_a=\partial_{t_a},\quad
J=t_1\partial_{t_2}-t_2\partial_{t_1},\quad
M=-i\frac{M_{\rm eff}}{\hbar},
\\
B_a&=r\partial_{t_a}-i\frac{M_{\rm eff}}{\hbar}\,t_a,\quad
D=2r\partial_r+t_a\partial_{t_a}+1,
\\
C&=r^2\partial_r+r\,t_a\partial_{t_a}+r
-\frac{iM_{\rm eff}}{2\hbar}(t_1^2+t_2^2),
\end{align}
\end{subequations}

with $a=1,2$ and repeated temporal indices summed. Their non-vanishing
Lie brackets are
\begin{subequations}
\begin{align}
[H_a,B_b]&=\delta_{ab}M, & [P,B_a]&=H_a,
\label{eq:comm-1}
\\
[J,H_a]&=-\epsilon_{ab}H_b, & [J,B_a]&=-\epsilon_{ab}B_b,
\label{eq:comm-2}
\\
[H_a,D]&=H_a, & [P,D]&=2P,
\label{eq:comm-3}
\\
[D,B_a]&=B_a, & [P,C]&=D,
\label{eq:comm-4}
\\
[D,C]&=2C, & [H_a,C]&=B_a,
\label{eq:comm-5}
\end{align}
\end{subequations}
where $\epsilon_{12}=1$.

This algebra is $\mathfrak{sch}(2)$, the Schr\"odinger algebra in two
spatial dimensions: $H_a$ are temporal translations, $B_a$ are
Carroll boosts, $J$ is temporal rotation, $P$ is evolution translation,
$D$ is dilatation, $C$ is special conformal, and $M$ is the central
extension.  The algebra has dimension 9; all Jacobi identities hold
by direct computation.  Setting $a=1$ and discarding $J$, $H_2$, $B_2$
recovers $\mathfrak{sch}(1)$ of Ref.~\cite{Najafizadeh2025-2}.

In the physically motivated temporal polar coordinates $t_a=(\tau\cos\psi,\tau\sin\psi)$,
$J=\partial_\psi$ and $D=2r\partial_r+\tau\partial_\tau+1$.
The symmetry generators of the original field $\widetilde\phi=r^{-1/2}\Phi$
are obtained by the similarity transformation
$\widetilde X=r^{-1/2}Xr^{1/2}$, shifting only the $r$-derivative
generators,
$\widetilde P=\partial_r+1/(2r)$,
$\widetilde D=2r\partial_r+\tau\partial_\tau+2$,
$\widetilde C=r^2\partial_r+r\tau\partial_\tau+3r/2-iM_{\rm eff}\tau^2/(2\hbar)$,
while $\widetilde H_a=H_a$, $\widetilde B_a=B_a$, $\widetilde J=J$,
$\widetilde M=M$.  Commutators are covariant under similarity
transformations, so the ungauged generators satisfy the same
$\mathfrak{sch}(2)$ algebra.

\subsection{Bargmann geometric interpretation}
\label{subsec:bargmann}

The $\mathfrak{sch}(2)$ symmetry and the central charge $M=-iM_{\rm eff}/\hbar$
can be realized geometrically through a Bargmann lift. Introduce the four-dimensional Bargmann space
with coordinates $(r,t_1,t_2,s)$ and null metric
\begin{equation}
d\Sigma^2=dt_1^2+dt_2^2+2\,dr\,ds.
\label{eq:bargmann-metric}
\end{equation}
The ansatz $\Psi(r,t_1,t_2,s)=e^{iM_{\rm eff}s/\hbar}\Phi(r,t_1,t_2)$
reduces the massless wave equation
$\Box_\Sigma\Psi=0$  with $\Box_\Sigma=\partial_{t_1}^2+\partial_{t_2}^2+2\partial_r\partial_s$ 
to
\begin{equation}
\Delta_t\Phi+\frac{2iM_{\rm eff}}{\hbar}\,\partial_r\Phi=0,
\label{eq:bargmann-reduction}
\end{equation}
which is exactly the gauge-fixed Carroll--Schr\"odinger equation.
The coordinate $s$ is the Bargmann fiber. Translation along this fiber,
$\partial_s$, is the central generator in the lifted Bargmann geometry. In the
fixed-mass sector, $\partial_s$ acts on the lifted field as $iM_{\rm eff}/\hbar$,
which is represented on the reduced wavefunction by the central charge
$M=-iM_{\rm eff}/\hbar$.

The lifted translations $H_a$, the rotation $J$, the radial translation $P$, the Galilean-type boosts $B_a$, and the fiber translation $\partial_s$ are Killing vector fields of the flat Bargmann metric. The dilation $D$ and special conformal generator $C$ are instead conformal Bargmann vector fields which preserve the Bargmann structure, and the fixed-mass wave equation, rather than the metric as ordinary isometries. With this distinction, the lifted brackets close into $\mathfrak{sch}(2)$, as in the standard Bargmann realization of Schr\"odinger symmetry~\cite{Duval2014}.
The Carroll--Schr\"odinger theory in Klein space identifies the Bargmann
``time'' with the spatial radius $r$, the Bargmann ``space'' with the
temporal two-plane $(t_1,t_2)$, and the Bargmann fiber with the
  post-Carrollian mass parameter $M_{\rm eff}=mc^3$.

The scalar momentum-coupled interaction $F(r,\tau,\psi)$ enters through the
Eisenhart--Bargmann deformation
\begin{equation}
d\Sigma_F^2=dt_1^2+dt_2^2+2\,dr\,ds-\frac{2F(r,t_1,t_2)}{M_{\rm eff}}\,dr^2.
\label{eq:eisenhart-metric}
\end{equation}
With the same fixed-mass ansatz, $\Box_{\Sigma_F}\Psi=0$ gives
$i\hbar\partial_r\Phi=[-\hbar^2\Delta_t/(2M_{\rm eff})+F]\Phi$. 

\subsection{Metaplectic structure of the temporal plane}
\label{subsec:metaplectic}

The quadratic sectors of the radial Carroll--Schr\"odinger equation
admit a common symplectic description.  Let
\begin{equation}
Z=(t_1,t_2,E_1,E_2)^T,
\qquad
E_a=-i\hbar\partial_{t_a},
\end{equation}
and introduce the standard symplectic matrix
\begin{equation}
\mathbb J=
\begin{pmatrix}
0&I_2\\
-I_2&0
\end{pmatrix}.
\end{equation}
For a quadratic Hamiltonian on the temporal phase space,
\begin{equation}
H(r)=\frac12 Z^T G(r)Z,
\label{eq:quadratic-temporal-H}
\end{equation}
the associated classical equations are linear,
\begin{equation}
\frac{dZ}{dr}=\mathbb JG(r)Z.
\label{eq:symplectic-flow}
\end{equation}
Thus
\begin{equation}
Z(r)=S(r,r_0)Z(r_0),
\qquad
S(r,r_0)\in Sp(4,\mathbb R).
\label{eq:sp4-flow}
\end{equation}
The quantum radial evolution is the metaplectic lift of this symplectic
map.  We use the Heisenberg convention
\begin{equation}
U_S^\dagger ZU_S=S(r,r_0)Z,
\qquad
U_S\in Mp(4,\mathbb R).
\label{eq:metaplectic-lift}
\end{equation}
Equivalently, if the opposite state-action convention is used, the same
relation is written with $S^{-1}$.

For the free Hamiltonian,
\begin{equation}
H_0=\frac{E_1^2+E_2^2}{2M_{\rm eff}},
\end{equation}
the symplectic matrix is
\begin{equation}
S_{\rm free}(\Delta r)=
\begin{pmatrix}
I_2&\Delta r\,I_2/M_{\rm eff}\\
0&I_2
\end{pmatrix}.
\label{eq:S-free}
\end{equation}
Its metaplectic lift is precisely the propagator of
Sec.~\ref{sec:propagator}.  The Gaussian spreading law in
Sec.~\ref{subsec:gaussian} is the image of the covariance matrix of a
temporal Gaussian under the linear transformation
\eqref{eq:S-free}.  The temporal Rayleigh range
$r_*=M_{\rm eff}\sigma_0^2/\hbar$ is therefore the diffraction length
of a metaplectic shear on the temporal plane.

The temporal oscillator corresponds to a symplectic rotation.  The
Hamiltonian
\begin{equation}
H_\Omega=\frac{E_1^2+E_2^2}{2M_{\rm eff}}
+\frac12M_{\rm eff}\Omega^2(t_1^2+t_2^2)
\end{equation}
generates the block rotation
\begin{equation}
S_\Omega(\Delta r)=
\begin{pmatrix}
\cos(\Omega\Delta r)I_2&
\dfrac{\sin(\Omega\Delta r)}{M_{\rm eff}\Omega}I_2\\
-M_{\rm eff}\Omega\sin(\Omega\Delta r)I_2&
\cos(\Omega\Delta r)I_2
\end{pmatrix}.
\label{eq:S-oscillator}
\end{equation}
The circular creation and annihilation operators introduced in
Sec.~\ref{subsec:oscillator} are the diagonal oscillator coordinates
for this metaplectic rotation.  The radial $SU(1,1)$ generators are the
quadratic, angular-momentum-preserving part of the same metaplectic
representation.

Within the effective temporal-connection model, the Landau-type and Fock--Darwin-type sectors are also quadratic. A
uniform temporal connection adds an antisymmetric part to the temporal
phase-space matrix $G(r)$.  Equivalently, the classical flow splits into
two chiral metaplectic rotations with frequencies
\begin{equation}
\Omega_\pm=\sqrt{\Omega^2+\frac{\omega_t^2}{4}}
\pm\frac{\omega_t}{2},
\end{equation}
as in the Fock--Darwin spectrum. Thus the same metaplectic language also organizes the effective quadratic connection sectors.

\section{Radial quantization and radial-ordered propagator}
\label{sec:propagator}

A first-order action yielding the gauge-fixed equation \eqref{eq:gauge-fixed-symmetry-start} is
\begin{equation}
S_\Phi=\int\!dr\,d^2t
\!\left[i\hbar\Phi^\dagger\partial_r\Phi
-\frac{\hbar^2}{2M_{\rm eff}}\partial_a\Phi^\dagger\partial_a\Phi\right].
\label{eq:action}
\end{equation}
Equal-$r$ canonical quantization gives
$[\widehat\Phi(r,\mathbf{t}),\widehat\Phi^\dagger(r,\mathbf{t}')]=\delta^{(2)}(\mathbf{t}-\mathbf{t}')$.
For the original field \eqref{eq:22CS-final},
$[\widehat{\widetilde\phi},\widehat{\widetilde\phi}^\dagger]=r^{-1}\delta^{(2)}(\mathbf{t}-\mathbf{t}')$.

The retarded (radial-ordered) Green kernel,
$\mathcal{S}_xG_{\rm ret}^\Phi=i\hbar\delta(r-r')\delta^{(2)}(\mathbf{t}-\mathbf{t}')$
with $G_{\rm ret}^\Phi=0$ for $r<r'$, is the exact two-dimensional
Schr\"odinger propagator,
\begin{equation}
G_{\rm ret}^\Phi(r,\mathbf{t};r',\mathbf{t}')
=\Theta(\Delta r)\frac{M_{\rm eff}}{2\pi i\hbar\,\Delta r}
\exp\!\left[\frac{iM_{\rm eff}|\mathbf{t}-\mathbf{t}'|^2}{2\hbar\,\Delta r}\right].
\label{eq:propagator}
\end{equation}
Now returning to the original ungauged field,
$G_{\rm ret}^{\widetilde\phi}=\sqrt{r'/r}\,G_{\rm ret}^\Phi$.
In polar temporal coordinates, using
$|\mathbf{t}-\mathbf{t}'|^2=\tau^2+\tau'^2-2\tau\tau'\cos(\psi-\psi')$,
the Jacobi--Anger expansion gives
\begin{align}
G_{\rm ret}^{\widetilde\phi}
={}&\Theta(\Delta r)\sqrt{\frac{r'}{r}}\,
\frac{M_{\rm eff}}{2\pi i\hbar\,\Delta r}\,
e^{i\gamma(\tau^2+\tau'^2)}
\nonumber\\
&\times\sum_{\ell\in\mathbb{Z}}(-i)^\ell J_\ell(z)\,e^{i\ell(\psi-\psi')},
\label{eq:propagator-polar}
\end{align}
where $\gamma=M_{\rm eff}/(2\hbar\,\Delta r)$ and
$z=M_{\rm eff}\tau\tau'/(\hbar\,\Delta r)$.  
Because the reduced equation is Schr\"odinger-like in the radial variable \(r\),
the radial-ordered prescription defines the natural radial initial-value problem even after returning to the original field. The
condition \(G_{\rm ret}=0\) for \(r<r'\) means that data on an equal-\(r\)
slice propagate only toward larger \(r\). For fixed \(\Delta r>0\), the kernel
is nonzero at arbitrary separation in the temporal plane, as expected for
Schr\"odinger-type spreading on the configuration space.

\section{Hamilton--Jacobi limit}
\label{sec:classical-limit}

Starting from the Carroll--Schr\"odinger equation~\eqref{eq:22CS-final}, substituting the WKB ansatz $\widetilde\phi=\mathcal Ae^{iS/\hbar},$ and taking the classical limit by keeping the terms of order $\hbar^0$, we obtain
\begin{equation}
\partial_rS+\frac{1}{2M_{\rm eff}}
\left[(\partial_\tau S)^2+\frac{(\partial_\psi S)^2}{\tau^2}\right]=0.
\label{eq:hj-temporal-plane}
\end{equation}
It is worth noting that the extra 1/(2r) term in Eq.~\eqref{eq:22CS-final} does not survive this limit, since it appears only at subleading order in $\hbar$. Separating $S=S_\tau(\tau)\pm J\psi-pr$ gives the post-Carrollian
momentum
\begin{equation}
P_{\rm PC}=\frac{1}{2}M_{\rm eff}\dot\tau^2+\frac{J^2}{2M_{\rm eff}\tau^2},
\qquad\dot\tau=d\tau/dr.
\label{eq:ppc-centrifugal}
\end{equation}
The post-Carrollian momentum \(P_{\rm PC}=mc^3/(2v^2)\), written in the form
\eqref{eq:ppc-centrifugal}, makes its interpretation more transparent. It has
the same structure as a Galilean kinetic energy on the temporal plane, with
\(r\) as the evolution parameter, \(M_{\rm eff}=mc^3\) as the effective mass,
and \(J\) giving the temporal centrifugal term. In this form, \(P_{\rm PC}\) is
seen directly as the generator of the reduced temporal dynamics. The same expression follows directly from the tachyonic mass shell. In the spatial carrier branch one has
\begin{equation}
P^2c^2=m^2c^4+E^2,
\qquad
E^2=E_\tau^2+\frac{J^2}{\tau^2}.
\end{equation}
For the post-Carrollian sector, expand the positive spatial momentum around the rapidly oscillating carrier,
\begin{equation}
P=mc\sqrt{1+\frac{E^2}{m^2c^4}}
=mc+\frac{E^2}{2mc^3}+\cdots .
\end{equation}
After subtracting the carrier momentum \(mc\), the finite post-Carrollian momentum is
\begin{equation}
P_{\rm PC}=\frac{E^2}{2mc^3}=\frac{E_\tau^2}{2M_{\rm eff}}+\frac{J^2}{2M_{\rm eff}\tau^2},
\label{eq:ppc-energy}
\end{equation}
where the second term is the centrifugal contribution absent in the 1 dimensional theory.

\subsection{Free trajectories, centrifugal barrier, and quantum correspondence}

Jacobi's theorem gives
\begin{equation}
r(\tau)=r_0+\sqrt{\frac{M_{\rm eff}}{2P_{\rm PC}}}\,
\sqrt{\tau^2-\tau_*^2},
\quad
\tau_*=\frac{|J|}{\sqrt{2M_{\rm eff}P_{\rm PC}}},
\label{eq:r-tau-trajectory}
\end{equation}
and $\psi(\tau)=\psi_0\mp\arcsin(\tau_*/\tau)$.
For $\tau<\tau_*$ the trajectory does not exist classically; $\tau_*$
behaves as the temporal turning radius.  The relation $\psi(\tau)=\psi_0\mp\arcsin(\tau_*/\tau)$
is equivalent to
\begin{equation}
\tau\sin(\psi-\psi_0)=\mathrm{const},
\label{eq:straight-line}
\end{equation}
which is a straight line in the Cartesian temporal plane $(t_1,t_2)$. Thus, even though the configuration space contains two temporal directions, a free classical trajectory selects a single straight temporal line by its initial data. The second temporal coordinate does not force any genuinely two-time wandering in the free sector;
conservation of $J$ implies that $\ d\psi /{d\tau}=J/(M_{\rm eff}\tau^2\dot\tau)$
accounts for the angular sweep needed to maintain rectilinearity.
The classical forbidden region $\tau<\tau_*$ has its quantum counterpart
in the Bessel node structure: $J_{|n|}(\lambda\tau)\sim\tau^{|n|}$
near $\tau=0$, so the wavefunction vanishes at the temporal origin
for $n\neq0$ and only the $n=0$ mode penetrates to $\tau=0$.  For
$\tau\gg\tau_*$, $r\simeq r_0+\sqrt{M_{\rm eff}/(2P_{\rm PC})}\,\tau$
and $\psi\to\psi_0$, recovering the one-dimensional dynamics~\cite{Rojas2025}.

\subsection{Interacting Hamilton equations}

With interaction momentum $F(r,\tau,\psi)$ and temporal-energy
shift $V(r,\tau,\psi)$, the radial Hamiltonian is
\begin{equation}
\mathcal{H}_r=F+\frac{(E-V)^2}{2M_{\rm eff}},
\quad
E=\sqrt{E_\tau^2+J^2/\tau^2},
\label{eq:interacting-H}
\end{equation}
with Hamilton equations ($\dot{f}=df/dr$)
\begin{subequations}
\begin{align}
\dot\tau&=\frac{E-V}{M_{\rm eff}}\frac{E_\tau}{E},
\quad
\dot\psi=\frac{E-V}{M_{\rm eff}}\frac{J}{\tau^2E},
\\
\dot E_\tau&=-\partial_\tau F+\frac{E-V}{M_{\rm eff}}
\!\left(\partial_\tau V+\frac{J^2}{\tau^3E}\right),
\\
\dot J&=-\partial_\psi F+\frac{E-V}{M_{\rm eff}}\partial_\psi V.
\end{align}
\end{subequations}
When $F$ and $V$ are $\psi$-independent, $J$ is conserved.  Setting
$J=0$ or taking $\tau\to\infty$ recovers the one-dimensional equations
of Refs.~\cite{Najafizadeh2025-1,Rojas2025}.

\section{Curved-Background Carrier Reduction}
\label{sec:gravitational-sector}

\subsection{Carrier congruence and local curved Carroll--Schr\"odinger equation}
\label{subsec:generalcurved}
For a probe tachyonic scalar $(\Box_g+\mu^2)\phi=0$ on a local
split-signature Einstein vacuum $R_{\mu\nu}[g]=0$, we introduce the
carrier $\phi=e^{-i\mu S}\widetilde\phi$, where the eikonal phase satisfies the unit spacelike condition
\begin{equation}
g^{\mu\nu}\nabla_\mu S\,\nabla_\nu S=1,
\label{eq:curved-eikonal-main}
\end{equation}
and defines the spacelike carrier congruence
\begin{equation}
V^\mu=g^{\mu\nu}\nabla_\nu S .
\label{eq:carrier-vector-main}
\end{equation}
Thus \(V^\mu V_\mu=1\). Substituting the factored-out field into the Klein--Gordon equation and canceling the leading mass term with Eq.~\eqref{eq:curved-eikonal-main}
gives
\begin{equation}
\Box_g\widetilde\phi
-2i\mu\left(V^\mu\nabla_\mu+\frac12\nabla_\mu V^\mu\right)\widetilde\phi=0 .
\label{eq:curved-envelope-before-limit-main}
\end{equation}
We now take the post-Carrollian contraction. Keeping fixed the generalized
post-Carrollian momentum derived in Appendix~\ref{app:curved-PC-momentum}
gives the same scaling as in the flat theory. The resulting local curved
Carroll--Schr\"odinger equation is
\begin{multline}
\frac{\hbar^2}{2M_{\rm eff}}\Delta_{\mathcal T}^{(g,V)}\widetilde\phi
+i\hbar\left(V^\mu\nabla_\mu+\frac12\nabla_\mu V^\mu\right)\widetilde\phi=0,
\\
M_{\rm eff}=mc^3 .
\label{eq:general-curved-CS-main}
\end{multline}
Here \(\Delta_{\mathcal T}^{(g,V)}\) denotes the leading inverse-metric
contraction of second derivatives along the two temporal directions selected by
the carrier branch. Appendix~\ref{app:general-curved-carrier} collects the general derivation, its
carrier-adapted Schr\"odinger form, and a brief discussion of the asymptotic
radial branch that recovers the flat equation with the \(i\hbar/(2r)\)
half-density term. As an illustrative example, the next
subsection specializes this construction to a particular split-signature vacuum
geometry.

\subsection{Illustrative example SO(2,1) parent symmetry}

A minimal requirement for a gravitational extension is asymptotic compatibility with the flat \(2+2\) theory: far from the gravitational source, the curved Carroll--Schr\"odinger equation must reduce to the flat equation derived in Sec.~\ref{sec:derivation}. This requirement selects backgrounds that approach flat Klein space asymptotically. There is, however, an additional issue that is special to multi-time geometries.
Evolution problems with more than one time direction are known to involve
nontrivial determinism and well-posedness constraints
\cite{FouresBruhat1952,ChoquetBruhatGeroch1969,FriedrichRendall2000,CraigWeinstein2009}. In the present construction the quantum theory is not
evolved in both temporal directions, but on equal-spatial slices. The spatial
coordinate therefore plays the role of the evolution parameter, while the
two-dimensional temporal plane becomes the configuration space. 
The \(SO(2,1)\)-symmetric sector provides a controlled illustrative background because, by the Kleinian analogue of Birkhoff's theorem \cite{EassonPezzelle2024}, it selects a unique asymptotically flat vacuum solution within that symmetry class, while remaining compatible with one-dimensional spatial post-Carrollian evolution in the presence of two timelike directions. Locally, define the fixed-angle hypersurface
\begin{equation}
\Sigma_\theta:\quad \theta=\mathrm{const} .
\end{equation}
Here \(\Sigma_\theta\) is introduced to exhibit the local \(SO(2,1)\)-symmetric block explicitly, while the more general curved Carroll--Schr\"odinger equation~\eqref{eq:general-curved-CS-main} is not tied to this particular presentation.
The corresponding local block is
\begin{equation}
\left.ds^2\right|_{\Sigma_\theta}
=
dx^{2}
-c^{2}dt_1^{2}
-c^{2}dt_2^{2},
\label{eq:reduced-flat-sector-sigma}
\end{equation}
with invariant
\begin{equation}
\varrho^{2}
=
x^{2}
-c^{2}(t_1^{2}+t_2^{2}) = x^2 - c^2 \tau^2.
\label{eq:SO21-invariant}
\end{equation}

The full parent space remains four-dimensional with split signature \((2,2)\); the \(SO(2,1)\) symmetry acts only on the local block containing one spatial direction and the two temporal directions, while the second spatial direction is a spectator. This distinction is important because pure \(2+1\) Einstein gravity has no local propagating gravitational degrees of freedom~\cite{DeserJackiwTHooft1984,Witten1988,Carlip2005}, whereas here the block is used only as a symmetry sector inside a four-dimensional Kleinian geometry.

The local block~\eqref{eq:reduced-flat-sector-sigma} is preserved by \(SO(2,1)\). The corresponding boosts \(K_a=ct_a\partial_x+\frac{x}{c}\partial_{t_a},  a=1,2,\) contract under \(c\to0\) to the Carroll boosts \(B_a=x\partial_{t_a}\), which are precisely the boosts appearing in the flat \(\mathfrak{sch}(2)\) algebra.

\subsection{\(SO(2,1)\) Kleinian Schwarzschild geometry}

In Klein space, for \(\varrho^2>0\), introduce hyperbolic coordinates in the reduced block by
\begin{equation}
x=\varrho\cosh\chi,
\qquad
c\tau=\varrho\sinh\chi,
\qquad
\varrho^2=x^2-c^2\tau^2 .
\label{eq:rho-chi}
\end{equation}
Then
\begin{equation}
dx^2-c^2d\tau^2-c^2\tau^2d\psi^2
=
d\varrho^2
-\varrho^2
\left(
d\chi^2+\sinh^2\chi\,d\psi^2
\right).
\label{eq:SO21-metric}
\end{equation}

The corresponding \(SO(2,1)\)-symmetric parent ansatz may be written as
\begin{equation}
ds^2
=
n(\zeta,\varrho)c^2d\zeta^2
+
p(\zeta,\varrho)d\varrho^2
-
\varrho^2
\left(
d\chi^2+\sinh^2\chi\,d\psi^2
\right),
\label{eq:SO21-parent-ansatz}
\end{equation}
where \(y=c\zeta\) is the spectator spatial direction not acted on by the
\(SO(2,1)\) symmetry.  The vacuum sector is defined by
\begin{equation}
G_{\mu\nu}=0 .
\label{eq:vacuum-einstein-kleinian}
\end{equation}
Solving these equations in the \(SO(2,1)\)-symmetric sector gives the Kleinian
Schwarzschild metric \cite{EassonPezzelle2024},
\begin{equation}
ds^2
=
f(\varrho)c^2d\zeta^2
+
f(\varrho)^{-1}d\varrho^2
-
\varrho^2
\left(
d\chi^2+\sinh^2\chi\,d\psi^2
\right),
\label{eq:kleinian-schwarzschild}
\end{equation}
where
\begin{equation}
f(\varrho)=1-\frac{A}{\varrho}.
\end{equation}
Here \(A\) is the Kleinian Schwarzschild length scale.

\subsection{\(x\)--adapted carrier and effective equation}
\label{subsec:curved-CS}

For the probe tachyonic scalar field in the coordinates of Eq.~\eqref{eq:kleinian-schwarzschild}, the parent dynamics is described by \(\mathcal L
=
\frac{1}{2}\sqrt{|g|}
\left(
g^{\mu\nu}\partial_\mu\phi\,\partial_\nu\phi
-\mu^2\phi^2
\right)\), with \(\mu=mc/\hbar\), and \((\Box_g+\mu^2)\phi=0\). Since the post-Carrollian dynamics developed in this work is formulated as evolution along one spatial direction, our goal is to obtain an \(x\)-adapted carrier representative. Before the contraction, however, the natural \(SO(2,1)\)-invariant carrier is the \(\varrho\)-adapted one. Denoting its action phase by \(S_0^{(\varrho)}=mc\,S^{(\varrho)}\), the unit spacelike eikonal condition (\ref{eq:curved-eikonal-main}) gives \(\partial_\varrho S_0^{(\varrho)}=mc/\sqrt{f(\varrho)}\). We implement the post-Carrollian contraction by scaling the temporal block, equivalently \(c\tau\to\epsilon c\tau\), so that \(\varrho_\epsilon
=
\sqrt{x^2-\epsilon^2c^2\tau^2}
=
x-\epsilon^2c^2\tau^2/(2x)
+\mathcal O(\epsilon^4)\), and
\begin{equation}
f(\varrho_\epsilon)
=
1-\frac{A}{x}
+\mathcal O(\epsilon^2).
\label{eq:rho-epsilon-carrier}
\end{equation}
Keeping finite the generalized post-Carrollian momentum derived in Appendix~\ref{app:curved-PC-momentum} gives the same scaling as in the flat theory, \( \mu\to \mu/ \epsilon^2\), \(  M_{\rm eff}=mc^3 \). Equivalently, the action phase appearing in the exponential is rescaled as \(S_{0,\epsilon}=S_0/\epsilon^2\). Therefore the \(\varrho\)-adapted and \(x\)-adapted phases differ by the finite limit
\begin{equation}
\frac{
S_{0,\epsilon}^{(\varrho)}(\varrho_\epsilon)
-
S_{0,\epsilon}^{(\varrho)}(x)
}{\hbar}
=
-\frac{M_{\rm eff}\tau^2}{2\hbar\,x\sqrt{f(x)}}
+\mathcal O(\epsilon^2).
\label{eq:carrier-phase-diff}
\end{equation}
Thus the \(x\)-adapted carrier used below is the local post-Carrollian representative of the original \(SO(2,1)\)-invariant \(\varrho\)-adapted branch, with the envelope differing only by this finite temporal phase.

We now write this representative as
\begin{equation}
\phi=e^{-i\mu S}\widetilde\phi
=
e^{-iS_0/\hbar}\widetilde\phi,
\quad
S_0=mc\,S,
\quad
V^\mu=g^{\mu\nu}\nabla_\nu S .
\label{eq:curved-carrier-redefinition}
\end{equation}
The corresponding unit eikonal condition on the chosen local \(x\)-branch, \(S=S(x)\), gives
\begin{equation}
\partial_x S=\frac{1}{\sqrt{f(x)}},
\qquad
\partial_\tau S=0
\label{eq:parent-carrier-momentum}
\end{equation}
to the order relevant in the post-Carrollian limit. Substitution into the tachyonic Klein--Gordon equation gives
\begin{equation}
\Box_g\widetilde\phi
-
2i\mu
\left(
V^\mu\nabla_\mu+\frac12\nabla_\mu V^\mu
\right)\widetilde\phi
=0 .
\label{eq:carrier-reduced-curved-kg}
\end{equation}

On the local fixed-angle branch one then finds
\begin{equation}
V^\mu\nabla_\mu
=
\sqrt{f(x)}\,\partial_x
+
\frac{f(x)-1}{\sqrt{f(x)}}\frac{\tau}{x}\partial_\tau,
\qquad
\nabla_\mu V^\mu=\mathcal D_x,
\label{eq:local-branch-V-divergence}
\end{equation}
with
\begin{equation}
\mathcal D_x
=
\partial_x\left(\sqrt{f(x)}\right)
+
\frac{2\left[f(x)-1\right]}{x\sqrt{f(x)}}.
\label{eq:Dx-def}
\end{equation}
 Multiplying Eq.~\eqref{eq:carrier-reduced-curved-kg} by \(-\epsilon^2\) and taking the \(\epsilon\to0\) limit yields the local \(SO(2,1)\)-adapted Carroll--Schr\"odinger equation
\begin{multline}
0=
\frac{\hbar^2}{2M_{\rm eff}}
\left(
\partial_\tau^2
+
\frac{1}{\tau}\partial_\tau
+
\frac{1}{\tau^2}\partial_\psi^2
\right)\widetilde\phi
\\
+i\hbar
\left[
\sqrt{f(x)}\,\partial_x
+
\frac{f(x)-1}{\sqrt{f(x)}}\frac{\tau}{x}\partial_\tau
\right]\widetilde\phi
+
\frac{i\hbar}{2}\mathcal D_x\,\widetilde\phi.
\label{eq:curved-CS-expanded-x}
\end{multline}
For \(f(x)\to1\), \(\mathcal D_x\to0\), and Eq.~\eqref{eq:curved-CS-expanded-x} reduces to
\begin{equation}
\frac{\hbar^2}{2M_{\rm eff}}\nabla_t^2\widetilde\phi
+
i\hbar
\partial_x\widetilde\phi
=
0,
\label{eq:curved-flat-limit-x}
\end{equation}
which is the flat Carroll--Schr\"odinger equation in the local fixed-angle patch, with \(x\) as the spatial evolution coordinate.
This form follows from the local equal-\(\theta\) branch used in the \(SO(2,1)\)-adapted construction, where the spatial evolution is directional. In the radial polar representative, the reduced equation carries the half-density term \(i\hbar/(2r)\), which may be removed by the standard redefinition \(\widetilde\phi=r^{-1/2}\Phi\). The radial and \(x\)-adapted representatives are different carrier realizations of the same post-Carrollian contraction, using a radial carrier in the former case and an \(x\)-adapted carrier in the latter. The common Carroll--Schr\"odinger structure underlying these representatives is the one analyzed throughout this work via the symmetry algebra, propagator structure, continuity equation, and explicit solution sectors. The asymptotic radial limit leading to this one-dimensional representative is discussed in Appendix~\ref{app:general-curved-carrier}. The corresponding asymptotic limit may impose additional restrictions on the allowed Einstein vacuum geometries, while the Kleinian Schwarzschild branch considered here is sufficient to exhibit the effective gravitational Carroll--Schr\"odinger dynamics in the asymptotic region \(x\gg A\).

\subsection{Classical limit and exact integrated trajectories}
\label{subsec:curved-classical}

The WKB limit at leading order in $\hbar$ gives the curved
Hamilton--Jacobi equation
\begin{equation}
\sqrt{f}\,\partial_xS+\frac{f-1}{\sqrt{f}}\frac{\tau}{x}\partial_\tau S
+\frac{1}{2M_{\rm eff}}\!\left[(\partial_\tau S)^2+\frac{(\partial_\psi S)^2}{\tau^2}\right]=0,
\label{eq:curved-HJ}
\end{equation}
with effective radial Hamiltonian
\begin{equation}
H_{\rm PC}=\frac{f-1}{f}\frac{\tau}{x}P_\tau
+\frac{1}{2M_{\rm eff}\sqrt{f(x)}}
\!\left(P_\tau^2+\frac{J^2}{\tau^2}\right).
\label{eq:curved-Hamiltonian}
\end{equation}
The Hamilton equations ($\dot{g}=dg/dx$) are
\begin{subequations}
\begin{align}
\dot\tau&=\frac{f-1}{f}\frac{\tau}{x}+\frac{P_\tau}{M_{\rm eff}\sqrt{f}},
\label{eq:tau-eq}
\\
\dot\psi&=\frac{J}{M_{\rm eff}\sqrt{f}\,\tau^2},
\label{eq:psi-eq-grav}
\\
\dot P_\tau&=-\frac{f-1}{f}\frac{P_\tau}{x}
+\frac{J^2}{M_{\rm eff}\sqrt{f}\,\tau^3},
\label{eq:ptau-eq}
\\
\dot J&=0.
\label{eq:J-eq}
\end{align}
\end{subequations}
The gravitational field introduces a radial drift in $\tau$ proportional
to $(f-1)/f$ and rescales the temporal kinetic term by $f^{-1/2}$.
The angular momentum $J$ is conserved.  In the flat limit these
equations reduce to the free double-polar Hamilton equations.

For the Schwarzschild profile \(f(x)=1-A/x\), the Hamilton equations imply
\begin{equation}
\frac{d(f\tau)}{dx}
=
\frac{\sqrt{f}\,P_\tau}{M_{\rm eff}}.
\label{eq:d-ftau}
\end{equation}
On the temporal plane, with \(\mathbf{t}=(t_1,t_2)=\tau(\cos\psi,\sin\psi)\), this relation becomes
\begin{equation}
\frac{d}{dx}\!\left[f(x)\,\mathbf{t}(x)\right]
=
\frac{f(x)^{3/2}}{M_{\rm eff}f_i}\,\mathbf{P}_i,
\qquad
f_i\equiv f(x_i),
\label{eq:affine-vector}
\end{equation}
where  \( 
\mathbf{P}_i
=
P_{\tau,i}(\cos\psi_i,\sin\psi_i)
+
\frac{J}{\tau_i}(-\sin\psi_i,\cos\psi_i)
\). Integrating,
\begin{equation}
f(x)\,\mathbf{t}(x)
=f_i\,\mathbf{t}_i+\frac{I(x)}{M_{\rm eff}f_i}\,\mathbf{P}_i,
\qquad
I(x)=\int_{x_i}^x\!f(s)^{3/2}\,ds,
\label{eq:affine-solution}
\end{equation}
where $I(x)$ is the optical radial variable that replaces the
flat distance $x-x_i$ as the natural propagation parameter in the
curved Kleinian background.  In polar form, Eq.~\eqref{eq:affine-solution}
gives
\begin{equation}
\tau(x)=\frac{1}{f(x)}\!\left[
\!\left(f_i\tau_i+\frac{\Pi_i I}{M_{\rm eff}f_i}\right)^{\!2}
+\!\left(\frac{JI}{M_{\rm eff}f_i\tau_i}\right)^{\!2}
\right]^{1/2},
\label{eq:tau-curved}
\end{equation}
\begin{equation}
\tan[\psi(x)-\psi_i]=
\frac{JI/(M_{\rm eff}f_i\tau_i)}{f_i\tau_i+\Pi_i I/(M_{\rm eff}f_i)},
\label{eq:psi-curved}
\end{equation}
with $P_\tau(x)=\Pi_i f(x)/f_i$.

For $f(x)=1-A/x$, the integral $I(x)$ has the closed form
\begin{multline}
I(x)=(x+2A)\sqrt{f(x)}-3A\operatorname{arccosh}\sqrt{x/A}
\\
-\left[(x_i+2A)\sqrt{f(x_i)}-3A\operatorname{arccosh}\sqrt{x_i/A}\right].
\label{eq:I-exact}
\end{multline}
For $A/x\ll1$,
\begin{equation}
I(x)\approx(x-x_i)-\frac{3A}{2}\ln\frac{x}{x_i}+\mathcal{O}\!\left(\frac{A^2}{x}\right).
\label{eq:I-expansion}
\end{equation}

\subsubsection{Temporal-plane lensing}
\label{subsec:lensing}

Equation~\eqref{eq:affine-solution} shows that the Kleinian
Schwarzschild branch produces a lensing-type distortion on the temporal plane:
it replaces the flat distance $x-x_i$ by $I(x)/f_i$ and rescales the
observed temporal vector by $1/f(x)$.  The angular deflection of the
temporal trajectory, relative to the flat branch, follows by
differentiating Eq.~\eqref{eq:psi-curved} with respect to $I$.  To
first order in $A/x$, with $\Delta x=x_f-x_i$ and $f_i\simeq1$,
\begin{equation}
\delta\psi_f\simeq
-\frac{3A}{2}\ln\frac{x_f}{x_i}\,
\frac{J/M_{\rm eff}}{
\!\left(\tau_i+\dfrac{\Pi_i\Delta x}{M_{\rm eff}}\right)^{\!2}
+\!\left(\dfrac{J\Delta x}{M_{\rm eff}\tau_i}\right)^{\!2}}.
\label{eq:lensing-angle}
\end{equation}
This formula quantifies the angular
deviation of a temporal-plane trajectory caused by the Kleinian
gravitational source in the $x$-adapted directional branch. The invariant content of the branch result is the optical-path excess \(\Delta\mathcal I(x)=I(x)-(x-x_i)=-\frac{3A}{2}\ln(x/x_i)+\mathcal O(A^2/x)\), which converts the free temporal angular sweep into the deflection \eqref{eq:lensing-angle} through the conserved temporal angular momentum \(J\). The deflection vanishes for $J=0$ (radial
temporal trajectories are not deflected angularly) and grows
logarithmically in $x_f/x_i$ for large baselines.

In the flat limit $A\to0$, $I(x)\to x-x_i$, and
Eq.~\eqref{eq:affine-solution} reduces to $\mathbf{t}(x)=\mathbf{t}_i
+(x-x_i)\mathbf{P}_i/M_{\rm eff}$, which is the straight-line temporal
trajectory of Sec.~\ref{sec:classical-limit}.

\subsection{Optical variables and branch propagator}
\label{subsec:quantum-lensing}

We now use the same optical radial variable in the wave equation. In the
\(x\)-adapted branch, a dilation of the temporal plane together with a field
rescaling puts Eq.~\eqref{eq:curved-CS-expanded-x} in the flat form; this is
useful for writing the Green function.

\begin{proposition}[Optical form of the branch equation]
\label{prop:quantum-lensing}
For \(f(x)=1-A/x\), define
\begin{equation}
\begin{aligned}
\boldsymbol\xi&=f(x)\,\mathbf t,\\
I(x)&=\int^x f(s)^{3/2}\,ds,\\
\widetilde\phi&=f(x)^{3/4}\Phi .
\end{aligned}
\label{eq:exact-map-defs}
\end{equation}
Then the equation for \(\Phi(\boldsymbol\xi,I)\) becomes
\begin{equation}
i\hbar\,\partial_I\Phi
=-\frac{\hbar^2}{2M_{\rm eff}}\nabla_\xi^2\Phi,
\qquad
\nabla_\xi^2=\partial_{\xi_1}^2+\partial_{\xi_2}^2 .
\label{eq:free-in-optical}
\end{equation}
\end{proposition}

\begin{proof}
At fixed \(x\), the dilation \(\boldsymbol\xi=f\mathbf t\) gives
\(\nabla_t^2=f^2\nabla_\xi^2\). Writing
\(\sigma=f\tau=|\boldsymbol\xi|\), one also has
\begin{equation}
\partial_x\big|_t
=
\partial_x\big|_\xi+\frac{f'}{f}\sigma\partial_\sigma .
\label{eq:x-derivative-optical-bla}
\end{equation}
For \(f=1-A/x\), the first-order terms proportional to
\(\sigma\partial_\sigma\) cancel since \(f'=A/x^2\) and
\((f-1)/x=-A/x^2\). Now write the field rescaling more generally as
\(\widetilde\phi=f^p\Phi\). The zeroth-order terms coming from
\(\partial_x f^p\) and from the measure term in
Eq.~\eqref{eq:curved-CS-expanded-x} cancel for \(p=3/4\). With this choice, and
using \(dI/dx=f^{3/2}\), the branch equation reduces to
Eq.~\eqref{eq:free-in-optical}.
\end{proof}

The flat kernel associated with Eq.~\eqref{eq:free-in-optical} first gives the
evolution of the rescaled field \(\Phi\),
\begin{equation}
K_\Phi(\boldsymbol\xi,I;\boldsymbol\xi',I')
=
\Theta(\Delta I)
\frac{M_{\rm eff}}{2\pi i\hbar\,\Delta I}
\exp\!\left[
\frac{iM_{\rm eff}|\boldsymbol\xi-\boldsymbol\xi'|^2}
{2\hbar\,\Delta I}
\right],
\label{eq:optical-kernel-Phi}
\end{equation}
where \(\Delta I=I-I'\). Returning to the original variables of the branch gives
\begin{equation}
\widetilde\phi(\mathbf t,x)
=
\int d^2t'\,
G^{\widetilde\phi}_{\rm br}(\mathbf t,x;\mathbf t',x')
\,\widetilde\phi(\mathbf t',x'),
\label{eq:curved-propagator-evolution}
\end{equation}
with
\begin{multline}
G^{\widetilde\phi}_{\rm br}(\mathbf t,x;\mathbf t',x')
=
\Theta(\Delta I)\,
f(x)^{3/4}f(x')^{5/4}
\frac{M_{\rm eff}}{2\pi i\hbar\,\Delta I}
\\
\times
\exp\!\left[
\frac{iM_{\rm eff}|f(x)\mathbf t-f(x')\mathbf t'|^2}
{2\hbar\,\Delta I}
\right],
\label{eq:curved-propagator}
\end{multline}
where now \(\Delta I=I(x)-I(x')\).

\section{Discussion}
\label{sec:discussion}

The post-Carrollian contraction of the tachyonic Klein--Gordon equation in flat Klein space $\mathbb K^{2,2}$ leads to a Carroll--Schr\"odinger equation whose equal-radius configuration space is the temporal two-plane. An important point in the derivation is that the scaling $\mu\to\mu/\epsilon^2$ selects the sector in which the post-Carrollian momentum $P_{\rm PC}=E^2/(2mc^3)$ remains finite. This leaves $M_{\rm eff}=mc^3$ as the effective inertial parameter of the reduced theory, and the classical post-Carrollian momentum can be written as $P_{\rm PC}=\frac{1}{2}M_{\rm eff}\dot\tau^2+J^2/(2M_{\rm eff}\tau^2)$, with $\dot\tau=d\tau/dr$, which fits naturally with the temporal interpretation of the reduced motion. In the large-$\tau$ regime, the centrifugal term becomes negligible and the classical dynamics approach the one-dimensional post-Carrollian relations.

The additional temporal direction introduces a conserved temporal angular momentum, organizes the regular modes into vortex sectors, generates the centrifugal term in the Hamilton--Jacobi limit, and controls oscillator degeneracies, effective connection sectors, and the two-body relative holonomy. The reduced theory therefore realizes on the temporal configuration plane several structures familiar from planar quantum mechanics, including Bessel sectors, radial oscillator towers, Aharonov--Bohm holonomy, Landau-type curvature, and Fock--Darwin splitting. The connection sectors used here are effective couplings on this reduced plane, and a first-principles post-Carrollian electromagnetic construction remains a natural direction for future work.

Once the relative temporal coordinate lies in \(\mathbb R^2\setminus\{0\}\), inequivalent $U(1)$ quantizations become available. This gives an anyon-like temporal holonomy and a fractional exchange phase under half-monodromy of the relative coordinate, with the realized holonomy depending on the microscopic interaction or on the choice of core regularization.

The curved construction carries the same carrier idea beyond the flat background. The generalized post-Carrollian quantity is kept finite along a chosen spatial branch, leading to a branch-dependent reduced equation on split-signature Einstein backgrounds. In the illustrative $SO(2,1)$-symmetric Kleinian Schwarzschild sector, the selected $x$-adapted branch contains drift and measure terms induced by the geometry, and at the classical level the Kleinian gravitational source produces an angular deviation on the temporal plane.

A deeper treatment of the curved sector, together with a first-principles derivation of post-Carrollian gauge couplings, would be natural next steps. Within this framework, the $(2,2)$ extension of the Carroll--Schr\"odinger program brings together temporal angular momentum, vortex sectors, anyonic relative holonomy, and branch-dependent curved effects in a single reduced system.

\section*{Acknowledgements}

The authors thank Manuel Beato for useful discussions.

\appendix

\section{Scaling of the post-Carrollian mass parameter}
\label{app:pc-scaling}

We show here that the scaling \(\mu\to\mu/\epsilon^2\) follows from keeping the post-Carrollian momentum finite. Consider the family of equivalent Carroll contractions \( c\to c_\epsilon=\epsilon^\alpha c        \), \( \tau\to\tau_\epsilon=\epsilon^{1-\alpha}\tau         \), \(  m\to m_\epsilon=\epsilon^\beta m,  \) with fixed spatial coordinates. This preserves the physical scaling \( c_\epsilon\tau_\epsilon=\epsilon\,c\tau        \), so that different values of \(\alpha\) simply distribute the same contraction between \(c\) and the temporal coordinate.

Since \(v\sim dx/d\tau\), fixed spatial coordinates give \( v_\epsilon= dx/d\tau_\epsilon = \epsilon^{\alpha-1}v \). Therefore the post-Carrollian momentum scales as
\begin{equation}
P_{{\rm PC},\epsilon}
=
\frac{m_\epsilon c_\epsilon^3}{2v_\epsilon^2}
=
\epsilon^{\beta+\alpha+2}P_{\rm PC}.
\label{eq:ppc-scaling}
\end{equation}
Keeping \(P_{{\rm PC},\epsilon}\) finite and nonzero requires \( \beta=-\alpha-2  \).
It follows that the Klein--Gordon mass parameter scales as \( \mu_\epsilon
=
m_\epsilon c_\epsilon/\hbar
=
\epsilon^{\beta+\alpha}\mu
=
\epsilon^{-2}\mu     \). Thus \(\mu\to\mu/\epsilon^2\), independently of how the contraction is represented.

The convention used in the main text corresponds to \(\alpha=1\), namely \( c\to\epsilon c \), \( \tau\to\tau \), \(m\to\epsilon^{-3}m \), for which the invariant quantity \(M_{\rm eff}=mc^3\) is manifest. 

\section{Generalized carrier functional and radial integral form}
\label{app:generalized-potential}

We briefly summarize a useful way of viewing the radial carrier factor used in the post-Carrollian reduction in Sec. \ref{sec:derivation}. To preserve linear superposition, a general field redefinition \( \phi = F(\widetilde\phi, q_1, q_2)\) must be linear in the reduced field, so we write \(\phi=g(q_1,q_2;\mu)\widetilde\phi+b(q_1,q_2;\mu)\) and keep only the multiplicative carrier \(g\).

After imposing the Carrollian scaling \(c\tau\to\epsilon c\tau\), \(\mu\to\mu/\epsilon^2\), and defining \(\kappa:=\mu/\epsilon^2\), the leading carrier equation may be written as
\begin{equation}
\nabla^2 g+
\left[
\kappa^2+\frac{V_{\rm geom}(q_1,q_2)}{\epsilon^2}
\right]g=0.
\label{eq:g-master-short}
\end{equation}
For the rotationally symmetric case \(g=g(r)\), this becomes
\begin{equation}
\frac{1}{r}\frac{d}{dr}\left(r\frac{dg}{dr}\right)
+
\left[
\kappa^2+\frac{V_{\rm geom}(r)}{\epsilon^2}
\right]g(r)=0.
\label{eq:g-radial-short}
\end{equation}

The rescaled carrier used in the main text is \(g(r)=\dfrac{\epsilon}{\sqrt{\mu}}\exp(-i\mu r/\epsilon^2)=\dfrac{1}{\sqrt{\kappa}}e^{-i\kappa r}\). Substituting into Eq.~\eqref{eq:g-radial-short} gives \(V_{\rm geom}(r)=i\mu/r\), which is the term present in Eq.~\eqref{eq:22CS-final}. This should be understood as the effective measure term generated by the polar radial carrier.

For completeness, writing \(g(r)=f(r)/\sqrt r\), Eq.~\eqref{eq:g-radial-short} becomes
\begin{equation}
f''(r)+\kappa^2 f(r)
=
-
\left[
\frac{1}{4r^2}
+
\frac{V_{\rm geom}(r)}{\epsilon^2}
\right]f(r).
\label{eq:f-helmholtz-short}
\end{equation}
Hence \(f\) satisfies the integral equation
\begin{equation}
f(r)=f_0(r)-\int dr'\,G_\kappa(r,r')
\left[
\frac{1}{4r'^2}
+
\frac{V_{\rm geom}(r')}{\epsilon^2}
\right]f(r'),
\label{eq:f-integral-short}
\end{equation}
where \(f_0(r)=Ae^{i\kappa r}+Be^{-i\kappa r}\). For outgoing boundary conditions one may take \(G_\kappa(r,r')=e^{i\kappa|r-r'|}/(2i\kappa)\). This gives the corresponding integral representation of the carrier factor \(g(r)\).

\section{Details of the curved carrier reduction}
\label{app:general-curved-carrier}

This appendix records the curved-background carrier reduction used in the main text before specializing to the \(SO(2,1)\)-symmetric vacuum. Let \(g\) be a local split-signature Einstein metric, \(R_{\mu\nu}[g]=0\), and let the probe scalar obey \((\nabla_\nu\nabla^\nu+\mu^2)\phi=0\), with \(\mu=mc/\hbar\). We introduce the carrier ansatz \(\phi=e^{-i\mu S}\widetilde\phi\), where the eikonal phase satisfies \(g^{\mu\nu}\nabla_\mu S\nabla_\nu S=1\), and define the spacelike carrier congruence \(V^\mu=g^{\mu\nu}\nabla_\nu S\). Substitution gives
\begin{multline}
\Box_g\widetilde\phi
-2i\mu V^\mu\nabla_\mu\widetilde\phi
-i\mu(\nabla_\mu V^\mu)\widetilde\phi
\\
+\mu^2\left(1-g^{\mu\nu}\nabla_\mu S\nabla_\nu S\right)
\widetilde\phi=0.
\end{multline}
Using the eikonal condition, this reduces to
\begin{equation}
\Box_g\widetilde\phi
-2i\mu\left(V^\mu\nabla_\mu+\frac12\nabla_\mu V^\mu\right)\widetilde\phi=0 .
\label{eq:app-general-curved-envelope}
\end{equation}

We now take the post-Carrollian contraction. As in the flat theory, the limit is defined by keeping finite the generalized post-Carrollian momentum associated with the chosen carrier branch. In the representative used in the main text, \(c\to\epsilon c\) with the coordinates fixed; since \(P_{\rm PC}^{(g)}\) is proportional to the combination \(mc^3\), this gives \(M_{\rm eff}=mc^3=\text{fixed}\) and \(\mu\to\mu/\epsilon^2\), as in Appendices~\ref{app:pc-scaling} and~\ref{app:curved-PC-momentum}. Multiplying Eq.~\eqref{eq:app-general-curved-envelope} by \(-\epsilon^2\) and retaining the leading finite terms leaves the second-order operator acting on the temporal directions selected by the carrier branch. Denoting this temporal block by
\(
\Delta_{\mathcal T}^{(g,V)}
\),
one obtains
\begin{equation}
\frac{\hbar^2}{2M_{\rm eff}}
\Delta_{\mathcal T}^{(g,V)}\widetilde\phi
+i\hbar\left(
V^\mu\nabla_\mu+\frac12\nabla_\mu V^\mu
\right)\widetilde\phi=0.
\label{eq:app-general-curved-CS}
\end{equation}
Here \(V^\mu\) fixes the spatial evolution direction, while \(\Delta_{\mathcal T}^{(g,V)}\) is the temporal Laplacian of the reduced branch.

In a smooth carrier-adapted patch, choose coordinates \((\rho,y^a,z^A)\) such that \(V^\mu\partial_\mu=\partial_\rho\). If \(h\) is the determinant of the metric induced on \(\rho=\mathrm{const}\), then \(\nabla_\mu V^\mu=\partial_\rho\ln\sqrt{|h|}\), and Eq.~\eqref{eq:app-general-curved-CS} becomes
\begin{equation}
\frac{\hbar^2}{2M_{\rm eff}}\Delta_{\mathcal T}^{(g,V)}\widetilde\phi
+i\hbar\left(\partial_\rho+\frac12\partial_\rho\ln\sqrt{|h|}\right)\widetilde\phi=0 .
\label{eq:app-curved-CS-adapted}
\end{equation}
The half-density transformation \(\widetilde\phi=|h|^{-1/4}\Psi\) (which, in the double polar flat limit, reduces to the field redefinition used throughout the paper         
\(\widetilde\phi
=
\dfrac{1}{\sqrt r}\Phi \)) removes the measure term and gives the associated Schr\"odinger form 
\begin{equation}
i\hbar\partial_\rho\Psi
=
-\frac{\hbar^2}{2M_{\rm eff}}
|h|^{1/4}\Delta_{\mathcal T}^{(g,V)}|h|^{-1/4}\Psi .
\label{eq:app-curved-schrodinger-form}
\end{equation}

\subsection{Asymptotic radial limit}

Under asymptotically radial conditions, with \(V\to\partial_r\),
\(\sqrt{|h|}\sim c^2r\), and the temporal inverse metric approaching
the flat temporal plane, one has 
\begin{equation}
\nabla_\mu V^\mu\to\frac{1}{r},
\qquad
\Delta_{\mathcal T}^{(g,V)}\to\partial_{t_1}^2+\partial_{t_2}^2 .
\end{equation}
Thus, for Einstein-vacuum backgrounds whose asymptotic radial region satisfies these conditions, Eq.~\eqref{eq:app-general-curved-CS} reduces at leading order to
\begin{equation}
\frac{\hbar^2}{2M_{\rm eff}}(\partial_{t_1}^2+\partial_{t_2}^2)\widetilde\phi
+i\hbar\!\left(\partial_r+\frac{1}{2r}\right)\widetilde\phi=0,
\label{eq:app-asymptotic-radial}
\end{equation}
which is the flat \(2+2\) Carroll--Schr\"odinger equation~\eqref{eq:22CS-final},
with the \(i\hbar/(2r)\) half-density term arising from
\(\nabla_\mu V^\mu\to1/r\).

\section{Local post-Carrollian momentum in curved space}
\label{app:curved-PC-momentum}

For a tachyonic worldline in the block-diagonal metric
$ds^2=-c^2h_{AB}dq^Adq^B+\gamma_{ij}dx^idx^j$, the exact spatial
momentum is $P_i=mc\,\gamma_{ij}v^j/\sqrt{X-c^2T}$ with
$T=h_{AB}u^Au^B$ and $X=\gamma_{ij}v^iv^j$.  Expanding for
$c^2T/X\ll1$,
\begin{equation}
P_i=mc\,\widehat v_i+\frac{mc^3}{2}\frac{T}{X}\widehat v_i+ \cdots
\end{equation}
so the local post-Carrollian momentum is
\begin{equation}
(P_{\rm PC})_i=\frac{mc^3}{2}
\frac{h_{AB}u^Au^B\,\gamma_{ij}v^j}{(\gamma_{k\ell}v^kv^\ell)^{3/2}}.
\end{equation}
Under the post-Carrollian contraction, keeping $(P_{\rm PC})_i$ finite
requires the same scaling $\mu\to\mu/\epsilon^2$ as in the flat case.

\bibliography{References}

\end{document}